\documentclass[11pt]{article}

\usepackage{fullpage}
\usepackage{amsmath,amsthm,amsfonts,dsfont}
\usepackage{amssymb,latexsym,graphicx}
\usepackage{palatino}
\usepackage{mathpazo}
\usepackage{stmaryrd}
\usepackage{mathtools}
\usepackage{hyperref}
\usepackage[lined,boxed]{algorithm2e}
\usepackage{subfigure}
\usepackage{boxedminipage}
\usepackage{authblk}
\usepackage{caption}

\usepackage{geometry}
\newtheorem{theorem}{Theorem}[section]

\newtheorem{proposition}[theorem]{Proposition}
\newtheorem{lemma}[theorem]{Lemma}

\newtheorem{claim}[theorem]{Claim}

\newtheorem{corollary}[theorem]{Corollary}
\newtheorem{definition}[theorem]{Definition}

\newcommand{\R}{\ensuremath{\mathbb{R}}}
\newcommand{\Z}{\ensuremath{\mathbb{Z}}}

 \newcommand{\eps}{\varepsilon} 
\renewcommand{\epsilon}{\varepsilon}

\renewcommand{\vec}[1]{\ensuremath{\mathbf{#1}}}

\newcommand{\basis}{\ensuremath{\mathbf{B}}}
\newcommand{\problem}[1]{\mathrm{#1}}

\newcommand{\poly}{\mathrm{poly}}

\DeclareMathOperator*{\expect}{\mathbb{E}}

\newcommand{\scarequotes}[1]{``#1''}

\newcommand{\Q}{\mathbb{Q}}

\makeatletter
\def\imod#1{\allowbreak\mkern8mu({\operator@font mod}\,\,#1)}
\makeatother

\newcommand{\lat}{\mathcal{L}}
\DeclareMathOperator{\dist}{dist}
\DeclarePairedDelimiter\inner{\langle}{\rangle}
\DeclarePairedDelimiter\floor{\lfloor}{\rfloor}
\DeclarePairedDelimiter\ceil{\lceil}{\rceil}
\DeclarePairedDelimiter\length{\lVert}{\rVert}
\newif\iffull
\fulltrue 

\newcommand{\full}[2]{\iffull#1\else#2\fi} %

\begin{document}

\title{Discrete Gaussian Sampling Reduces to CVP and SVP}
\author[1]{
Noah Stephens-Davidowitz~\thanks{New York University}~\thanks{Work done while at the Simons Institute 2015 cryptography summer program. This material is based upon work partially supported by the National Science Foundation under Grant No.~CCF-1320188. Any opinions, findings, and conclusions or recommendations expressed in this material are those of the authors and do not necessarily reflect the views of the National Science Foundation.}\\
\texttt{noahsd@gmail.com}
}
\date{}
\maketitle

\begin{abstract}
The discrete Gaussian $D_{\lat - \vec{t}, s}$ is the distribution that assigns to each vector $\vec{x}$ in a shifted lattice $\lat - \vec{t}$ probability proportional to $e^{-\pi \length{\vec{x}}^2/s^2}$. It has long been an important tool in the study of lattices. More recently, algorithms for discrete Gaussian sampling (DGS) have found many applications in computer science. In particular, polynomial-time algorithms for DGS with very high parameters $s$ have found many uses in cryptography and in reductions between lattice problems. And, in the past year, Aggarwal, Dadush, Regev, and Stephens-Davidowitz showed $2^{n+o(n)}$-time algorithms for DGS with a much wider range of parameters and used them to obtain the current fastest known algorithms for the two most important lattice problems, the Shortest Vector Problem (SVP) and the Closest Vector Problem (CVP).

Motivated by its increasing importance, we investigate the complexity of DGS itself and its relationship to CVP and SVP. Our first result is a polynomial-time dimension-preserving reduction from DGS to CVP. There is a simple reduction from CVP to DGS, so this shows that DGS is equivalent to CVP. Our second result, which we find to be more surprising, is a polynomial-time dimension-preserving reduction from \emph{centered} DGS (the important special case when $\vec{t} = \vec0$) to SVP. In the other direction, there is a simple reduction from $\gamma$-approximate SVP for any $\gamma = \Omega(\sqrt{n/\log n})$, and we present some (relatively weak) evidence to suggest that this might be the best achievable approximation factor.

We also show that our CVP result extends to a much wider class of distributions and even to other norms.
\end{abstract}
\thispagestyle{empty}

\full{}{\newpage
\setcounter{page}{1}}

\section{Introduction}
A lattice $\lat \subset \Q^n$ is the set of all integer linear combinations of some linearly independent basis vectors $\vec{b}_1,\ldots, \vec{b}_n \in \Q^n$. 

The two central computational problems on lattices are the Shortest Vector Problem (SVP) and the Closest Vector Problem (CVP). Given a lattice $\lat \subset \Q^n$, the SVP is to find a shortest non-zero vector in $\lat$. Given a lattice $\lat \subset \Q^n$ and a target vector $\vec{t} \in \Q^n$, the CVP is to find a vector in $\lat$ whose distance to $\vec{t}$ is minimal. 

Algorithms for SVP and CVP, in both their exact and approximate versions, have found many diverse applications in computer science. They have been used to factor polynomials over the rationals~\cite{LLL82}, solve integer programming~\cite{Len83,Kan87,DPV11}, and break cryptographic schemes~\cite{Odl90,JS98,NS01}. And, over the past twenty years, a wide range of strong cryptographic primitives have been constructed with their security based on the \emph{worst-case} hardness of the approximate versions of these problems~\cite{Ajt96,MR07,GPV08, Gen09, Peikert09, Reg09,LPR10,BV11,BV14}.

Both problems are known to be hard, even to approximate to within the nearly polynomial factor of $n^{c/\log \log n}$ for some constant $c$~\cite{ABSS93,Ajt98,CN98,BS99,DKRS03,Mic01svp,Khot05svp,Micciancio12,HRsvp}. Indeed, CVP is in some sense \scarequotes{lattice complete} in that nearly all well-studied lattice problems are reducible to CVP via dimension-preserving (and approximation-factor-preserving) reductions. (See~\cite{Micciancio08} for a list of such problems.) In particular, a dimension-preserving reduction from SVP to CVP has long been known~\cite{GMSS99}. However, the best-known dimension-preserving reduction in the other direction only reduces $O(\sqrt{n})$-approximate CVP to SVP.

A powerful tool for studying lattices is the discrete Gaussian, the probability distribution $D_{\lat - \vec{t}, s}$ that assigns to each vector $\vec{x} \in \lat - \vec{t}$ probability proportional to its Gaussian mass, $e^{-\pi \length{\vec{x}}^2/s^2}$, for a lattice $\lat \subset \Q^n$, shift vector $\vec{t} \in \Q^n$, and parameter $s > 0$. The discrete Gaussian and the closely related theta functions have been used to prove transference theorems on lattices~\cite{banaszczyk, Cai03}; to show that $\sqrt{n}$-approximate CVP and SVP are in co-NP~\cite{AharonovR04}; to embed flat tori in a Hilbert space with low distortion~\cite{HavivR13}; to solve the Bounded Distance Decoding Problem~\cite{LiuLM06}; and even in the study of the Riemann zeta function (e.g., in~\cite{BianePY01}).

\begin{figure}[ht]
\begin{center}
\includegraphics[width=0.4 \textwidth]{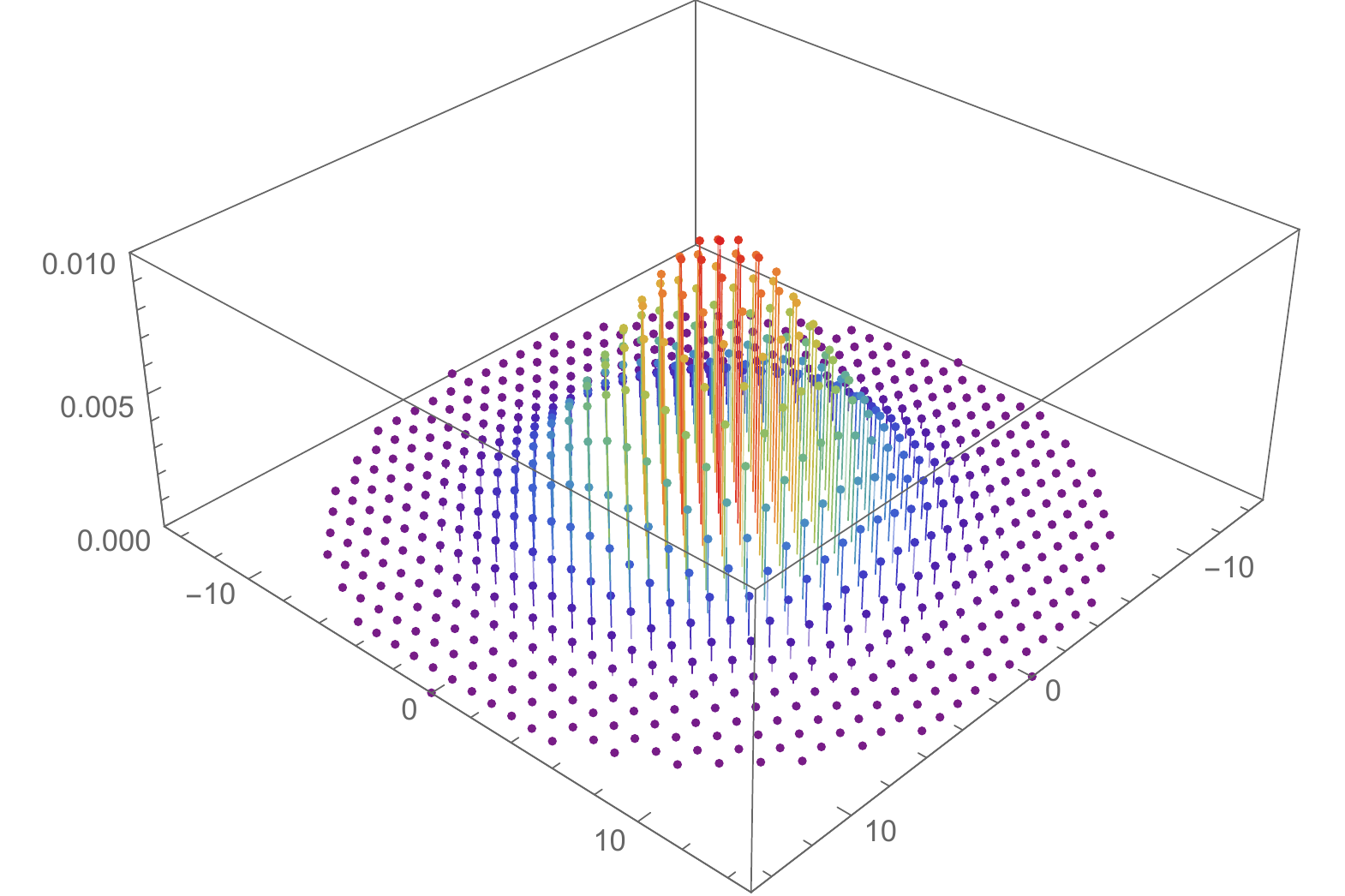}
\qquad
\includegraphics[width=0.4 \textwidth]{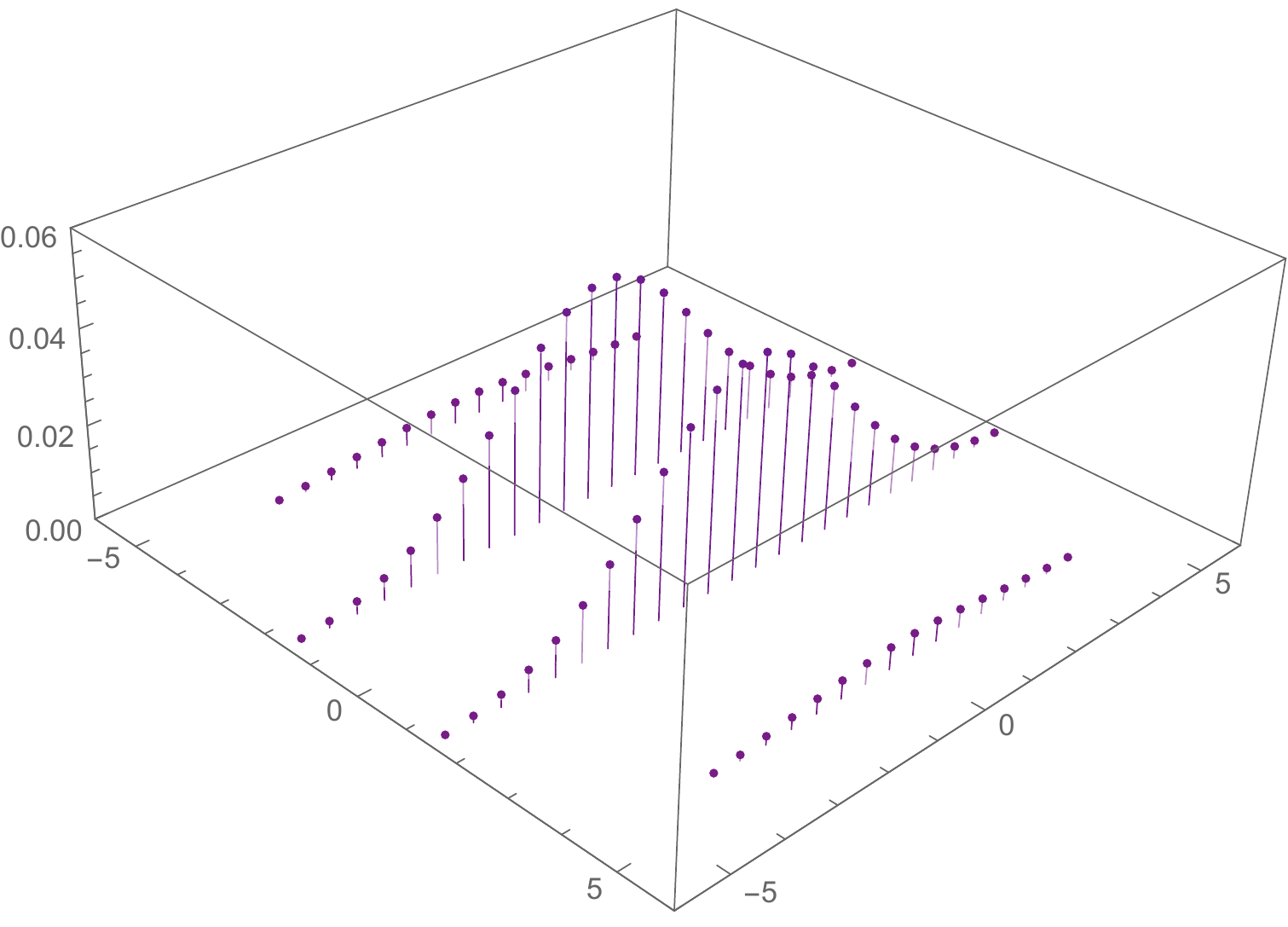}
\caption{\label{fig:DGS} 
Two very different discrete Gaussian distributions in two dimensions. On the left is $D_{\Z^2, 10}$. On the right is $D_{\lat - \vec{t},5}$, where $\lat$ is spanned by $3\vec{e}_1$ and $\vec{e}_2/2$, and $\vec{t} = 3\vec{e}_1/2 + \vec{e}_2/4$ is a \scarequotes{deep hole.}}
\end{center}
\end{figure}

Note that the discrete Gaussian is concentrated on relatively short vectors. In particular, in the important special case when the discrete Gaussian is \emph{centered} so that $\vec{t} = \vec0$, $D_{\lat, s}$ assigns higher weight to shorter lattice vectors. This suggests a connection between $D_{\lat, s}$ and SVP. In the more general case, $D_{\lat - \vec{t}, s}$ is concentrated on short vectors in the shifted lattice $\lat - \vec{t}$. By translating this distribution by $\vec{t}$ (i.e., considering the distribution of $D_{\lat - \vec{t}, s} + \vec{t}$), we obtain a distribution over the lattice that assigns higher weight to the vectors that are closest to $\vec{t}$, suggesting a connection between $D_{\lat - \vec{t}, s}$ and CVP. As the parameter $s$ becomes lower, the distribution becomes more concentrated. Indeed, one can show that samples from $D_{\lat - \vec{t}, s}$ (when suitably translated) yield $(1+\alpha \sqrt{n})$-approximate solutions to CVP when $s \approx \dist(\vec{t}, \lat)/\alpha$. (See Figure~\ref{fig:DGS} for two examples of the discrete Gaussian in two dimensions.)

Largely because of its connection to other lattice problems, algorithms for discrete Gaussian sampling (DGS) have recently played an important role in computer science.
Gentry, Peikert, and Vaikuntanathan introduced a polynomial-time trapdoor algorithm for sampling from the discrete Gaussian with very high parameters $s$ in order to construct a secure signature scheme \cite{GPV08}. 
And, many reductions between lattice problems use a DGS algorithm as a subroutine~\cite{Reg09,Pei10, MicciancioP13,BLPRS13}. But, these reductions also only work for very high parameters $s$. In particular, all previously known polynomial-time algorithms (even those with access to trapdoors and oracles) can only sample from $D_{\lat - \vec{t},s}$ when $s$ is significantly above the \scarequotes{smoothing parameter} of the lattice, in which case the discrete Gaussian \scarequotes{looks like a continuous Gaussian distribution} in a certain precise sense that we do not define here. (See \cite{MR07} for the formal definition.)

In the past year, Aggarwal, Dadush, Regev, and Stephens-Davidowitz introduced an exponential-time algorithm for sampling from the discrete Gaussian with much lower parameters in order to solve exact SVP~\cite{ADRS15}, and \cite{ADS15} showed how to extend this result to CVP. These are the current fastest-known algorithms for SVP and CVP. In particular,~\cite{ADRS15} showed how to sample exponentially many vectors from the \emph{centered} discrete Gaussian for \emph{any} parameter $s$ in $2^{n+o(n)}$ time, which yields a solution to SVP. \cite{ADS15} extended this work to show how to sample many vectors from $D_{\lat - \vec{t}, s}$ for very small parameters $s \approx \dist(\vec{t}, \lat)/2^{n}$, also in $2^{n+o(n)}$ time. Surprisingly, they showed how to use such an algorithm to construct a $2^{n+o(n)}$-time algorithm for CVP.\full{\footnote{It is easy to see that a discrete Gaussian sampler that works for any $\vec{t}$ and any $s$ is sufficient to solve CVP efficient. (We include a proof in Section~\ref{sec:CVPtoDGS} for completeness.) The difficulty in \cite{ADS15} is that the sampler only works for parameters $s$ greater than roughly $\dist(\vec{t}, \lat)/2^{n}$. While this minimum value is very small, this does not seem to be enough to efficiently solve exact CVP on its own. \cite{ADS15} manage to solve exact CVP in spite of this difficulty because their DGS algorithm outputs very many samples, which they use to recursively find an exact closest vector.}}{} (In Table~\ref{tab:DGS}, we summarize the previous known algorithms for discrete Gaussian sampling, together with the results of this work.)

\begin{table}
\begin{center}
\begin{tabular}{l l l  l  l }
& Shift & Parameter & Time & Notes \\
\hline
\hline
& Any $\vec{t}$ & $s \geq 2^{n \log n/\log \log n } \cdot \lambda_n$  & $\poly(n)$ & \cite{AKS01,GPV08} \\
& Any $\vec{t}$ & $s \geq \gamma \sqrt{\log n} \cdot \lambda_n$ & -- & Reduces to $\gamma$-approx. SVP \cite{GPV08, BLPRS13}.\\
& Any $\vec{t}$ & $s \gg \sqrt{n} \cdot \eta$ & -- & Quantum reduction to LWE \cite{Reg09}.\\
& Any $\vec{t}$ &   $s \geq \sqrt{2} \cdot \eta$ & $2^{n/2 + o(n)}$ & Outputs $2^{n/2}$ samples~\cite{ADRS15}.\\
& Any $\vec{t}$ & $s \gtrsim 2^{-n/\log n}\dist(\vec{t}, \lat)$ & $2^{n+o(n)}$ & Outputs many samples~\cite{ADS15}.\\
* & Any $\vec{t}$ &  Any $s$ & -- & Equivalent to CVP.\\
* & Any $\vec{t}$ &  Any $s$ & $2^{n+o(n)}$ & Follows from equivalence and \cite{ADS15}. \\
   \hline\hline
& $\vec{t} = \vec0$ &   Any $s$ & $2^{n+o(n)}$ & Outputs $2^{n/2}$ samples~\cite{ADRS15}.\\
* & $\vec{t} = \vec0$ &   Any $s$ & -- & Reduces to SVP.\\
   \hline
  \end{tabular}
    \caption{\label{tab:DGS}Known results concerning the problem of sampling from $D_{\lat - \vec{t}, s}$. Lines marked with a * are new results. We have omitted some constants. $\eta$ is the smoothing parameter, as defined in~\cite{MR07}, and $\lambda_n$ is the $n$th successive minimum. (They are related by $\eta/\sqrt{\log n} \lesssim \lambda_n \lesssim \sqrt{n} \cdot \eta$, where the upper bound is tight for the lattices that are relevant for cryptography. We also have have $\dist(\vec{t}, \lat) \leq \sqrt{n} \lambda_n/2$.)}
  \end{center}
\end{table}

All of these results reflect the increasing prominence of discrete Gaussian sampling algorithms in computer science. However, they left open a natural question: what is the complexity of DGS itself?  In particular, prior to this work, DGS was one of the only prominent lattice problems not known to be reducible to CVP via a dimension-preserving reduction. \full{(Another important example is the Lattice Isomorphism Problem.) }{}In fact, previously, there was simply no known algorithm that sampled from $D_{\lat - \vec{t}, s}$ for an arbitrary shift $\vec{t}$ and parameter $s > 0$, and it was not even known whether sampling from the \emph{centered} distribution $D_{\lat, s}$ could be reduced to a problem in NP. (Since DGS is a sampling problem, it technically cannot be placed directly in classes of decision problems or search problems like NP or FNP. But, we can still reduce it to such problems. See, e.g.,~\cite{Aaronson14} for a discussion of the complexity of sampling problems and their relationship to search problems.)

\subsection{Our results}

Our first main result is a dimension-preserving reduction from discrete Gaussian sampling to CVP. (See Theorem~\ref{thm:DGStoCVP}.) This immediately implies two important corollaries. First, together with the relatively straightforward reduction from CVP to DGS\full{ (see Section~\ref{sec:CVPtoDGS})}{}, this shows that CVP and DGS are equivalent via efficient dimension-preserving reductions. In particular, this suggests that the approach of \cite{ADS15} is in some (weak) sense the \scarequotes{correct} way to attack CVP, since we now know that any faster algorithm for CVP necessarily implies a similarly efficient discrete Gaussian sampler, and vice versa. Second, together with the result of \cite{ADS15}, this gives a $2^{n+o(n)}$-time algorithm for discrete Gaussian sampling that works for any parameter $s$ and shift $\vec{t}$, the first known algorithm for this problem.

Our second main result is a dimension-preserving reduction from \emph{centered} DGS to SVP. (See Theorem~\ref{thm:DGStoSVP}.) As we describe below, this result requires quite a bit more work, and we consider it to be more surprising, since, in a fixed dimension, an SVP oracle seems to be significantly weaker than a CVP oracle. In contrast to the CVP case, we know of no efficient reduction from SVP to centered DGS, and we do not even know whether centered DGS is NP-hard. (While \cite{ADRS15} use centered DGS to solve SVP, they require exponentially many samples to do so.) \full{We}{In the full version, we} present only a much weaker reduction from $\gamma$-approximate SVP to centered DGS for any $\gamma = \Omega(\sqrt{n/\log n})$. We also show that, for any $\gamma = o(\sqrt{n/\log n})$, 
no \scarequotes{simple} reduction from $\gamma$-SVP to centered DGS will work. (See Section~\ref{sec:SVPtoDGS}.)

Finally, we note that our proofs do not make use of any unique properties of the discrete Gaussian or of the $\ell_2$ norm. We therefore show a much more general result: any distribution that is close to a weighted combination of uniform distributions over balls in some norm reduces to CVP in this norm. (See Section~\ref{sec:other}.) In particular, sampling from the natural $\ell_q$ analogue of the discrete Gaussian is equivalent to CVP in the $\ell_q$ norm, under efficient dimension-preserving reductions. We imagine that a similar result holds for SVP, but since we know of no application, we do not endeavor to prove such a result in the more difficult setting of SVP.

\subsection{Proof overview}

\full{We now provide a high-level description of our techniques.}{}

\paragraph{Reduction from DGS to CVP. } 
Our basic idea is to sample from the discrete Gaussian $D_{\lat - \vec{t}, s}$ in two natural steps. We first sample some radius $r$ from a carefully chosen distribution. We then sample a uniformly random point in $(\lat - \vec{t}) \cap r B_2^n$. In particular, the distribution on the radius should assign probability to each radius $r$ that is roughly proportional to $e^{-\pi r^2/s^2} \cdot |(\lat - \vec{t}) \cap rB_2^n| $. (See the proof of Theorem~\ref{thm:DGStoCVP} for the exact distribution.) So, in order to solve DGS, it suffices to (1) compute $|(\lat - \vec{t}) \cap r B_2^n|$ for arbitrary $r$, and (2) sample a uniformly random point from $(\lat - \vec{t}) \cap r B_2^n$.

We actually use the same technical tool to solve both problems: lattice sparsification, as introduced by Khot~\cite{Khot05svp} (though our analysis is more similar to that of Dadush and Kun~\cite{DK13} and \cite{cvpp}). Intuitively, sparsification allows us to sample a random sublattice $\lat' \subset \lat$ of index $p$ such that for any vector $\vec{x} \in \lat$, we have $\Pr[\vec{x} \in \lat'] \approx 1/p$. Suppose we could find a sublattice $\lat'$ such that for the closest $N \approx p$ points to $\vec{t}$ in $\lat$, we have $\Pr[\vec{x} \in \lat'] = 1/p$, independently of the other points. Then, this would suffice for our two use cases. In particular, if the lattice has $N$ points in the ball of a given radius around $\vec{t}$, then $\lat' - \vec{t}$ would have a point in this ball with probability very close to $N/p$. We can use a CVP oracle to approximate this probability empirically, and we therefore obtain a good approximation for the number of lattice points in any ball. (\full{We achieve an approximation factor of $1+1/f(n)$ for any $f(n) = \poly(n)$. }{}See Theorem~\ref{thm:counter}.) Similarly, if we know that the number of lattice points in a ball of radius $r$ around $\vec{t}$ is roughly $N$, then we can take $p = \poly(n) \cdot N$ and repeatedly sample $\lat'$ until $\lat'$ has a point inside the ball of radius $r$ around $\vec{t}$. The resulting point will be a nearly uniformly random sample from the lattice points in the ball of radius $r$ around $\vec{t}$. Combining these two operations allows us to sample from the discrete Gaussian using a CVP oracle, as described above. (See Theorem~\ref{thm:DGStoCVP}.)

Unfortunately, sparsification does not give us exactly this distribution. More specifically, sparsification works as follows. Given a prime $p$ and lattice basis $\basis$, we sample $\vec{z} \in \Z_p^n$ uniformly at random and define the corresponding sparsified sublattice as 
\begin{equation}
\label{eq:sparsesublattice}
\lat' := \{\vec{x} \in \lat\ : \  \inner{\vec{z}, \basis^{-1}\vec{x}} \equiv 0 \bmod p \}
\; .
\end{equation}  
Then, for any vector $\vec{x} \in \lat$, we have $\Pr[\vec{x} \in \lat'] = 1/p$ unless $\vec{x} \in p\lat$ (in which case $\vec{x}$ is always in $\lat'$). Unfortunately, even if we ignore the issue that points in $p\lat$ do not behave properly, it is easy to see that these probabilities are not at all independent. For example, if $\vec{x} = \alpha \vec{y}$, then $\vec{x} \in \lat'$ if and only if $\vec{y} \in \lat'$. \full{And of course, more complex dependencies can exist as well. }{}Fortunately, we can get around this by using an idea from \cite{cvpp} (and implicit in \cite{DK13}). In particular, we can show that the probabilities are close to independent if we also shift the sublattice $\lat'$ by a \scarequotes{random lattice vector} $\vec{w} \in \lat$. I.e., while the distribution of the points in $\lat' \cap (rB_2^n + \vec{t})$ might be very complicated, each point in $\lat \cap (rB_2^n + \vec{t})$ will land in $\lat' - \vec{w}$ with probability $\approx 1/p$, and their distributions are nearly independent. (See Theorem~\ref{thm:shiftedsparsification} for the precise statement.) Our CVP oracle makes no distinction between lattices and shifted lattices (we can just shift $\vec{t}$ by $\vec{w}$), so this solution suffices for our purposes.

\paragraph{Reduction from centered DGS to SVP. } Our reduction from centered DGS to SVP uses the same high-level ideas described above, but the details are a bit more complicated. As in the CVP case, our primary tool is lattice sparsification, in which we choose a sparsified sublattice as in Eq.~\eqref{eq:sparsesublattice}. As before, we wish to control the distribution of the shortest vector in $\lat'$, and we note that, ignoring degenerate cases, $\vec{x}$ is a shortest vector of $\lat'$ if and only if $\vec{x} \in \lat'$ and $\vec{y}_1, \ldots, \vec{y}_N \notin \lat'$ where the $\vec{y}_i \in \lat$ are the non-zero lattice vectors shorter than $\vec{x}$ (up to sign). However, as in the CVP case, this probability can be affected by linear dependencies. In the CVP case, we solved this problem by considering a random shift of $\lat'$. But, this solution clearly does not work here because an SVP oracle simply \scarequotes{cannot handle} shifted lattices. We therefore have to deal explicitly with these dependencies.

The most obvious type of dependency occurs when $\vec{x}$ is not \emph{primitive}, so that $\vec{x} = \alpha \vec{y}_i$ for $|\alpha| > 1$. In this case, there is nothing that we can do---$\vec{y}_i$ is shorter than $\vec{x}$ and $\vec{y}_i \in \lat'$ if and only if $\vec{x} \in \lat'$, so $\vec{x}$ will \emph{never} be a shortest non-zero vector in $\lat'$. We therefore are forced to work with only primitive vectors (i.e., lattice vectors that are not a scalar multiple of a shorter lattice vector). Even if we only consider primitive vectors, it can still be the case that two such vectors are scalar multiples of each other mod $p$, $\vec{x} \equiv \alpha \vec{y}_i \bmod p\lat$. \full{Luckily, we show that this can only happen if there are $\widetilde{\Omega}(p)$ primitive vectors shorter than $\vec{x}$ in the lattice, so that this issue does not affect the $\widetilde{\Omega}(p)$ shortest primitive vectors. (See Lemma~\ref{lem:nogoodnameforthislemma}.) We also show that higher-order dependencies (e.g., equations of the form $\vec{x} \equiv \alpha \vec{y}_i + \beta \vec{y}_j \bmod p\lat$) have little effect. (See Lemma~\ref{lem:almostindependent}.)}{In the full version, we show that such issues can be overcome.} So, the shortest non-zero vector in the sparsified lattice will be distributed nearly uniformly over the $\widetilde{\Omega}(p)$ shortest primitive vectors in the original lattice. (See Theorem~\ref{thm:sparsification} and Proposition~\ref{prop:sparsifier} for the precise statement, which might be useful in future work.) 

As in the CVP case, this suffices for our purposes. In particular, if there are $N$ \emph{primitive} lattice vectors in the ball of radius $r$ centered at the origin for $N \leq \tilde{O}(p)$, then there will be a non-zero vector in $\lat' \cap rB_2^n$ with probability very close to $N/p$. With an SVP oracle, we can estimate this probability, and this allows us to approximate the number of primitive lattice vectors in a ball with very good accuracy. (See Theorem~\ref{thm:primcounter}.) And, the sparsification algorithm and SVP oracle also allow us to sample a primitive lattice vector in the ball of radius $r$ around the origin with nearly uniform probability, as in the CVP case. (See Lemma~\ref{lem:uniformsampler}.) 

Then, the same approach as before would allow us to use an SVP oracle to sample from the discrete Gaussian over the \emph{primitive} lattice vectors. In order to obtain the true discrete Gaussian, we first \scarequotes{add $\vec0$ in} by estimating the total Gaussian mass $\rho_s(\lat)$ and returning $\vec0$ with probability $1/\rho_s(\lat)$. Second, after sampling a primitive vector $\vec{x}$ using roughly the above idea, we sample an integer coefficient $z \in \Z \setminus \{ 0\}$ according to a one-dimensional discrete Gaussian (using an algorithm introduced by \cite{BLPRS13}) and output $z \vec{x}$. If we choose the primitive vector appropriately, we show that the resulting distribution is $D_{\lat, s}$.\full{\footnote{Interestingly, the problem of sampling from the centered discrete Gaussian over the \emph{primitive} lattice vectors, or even just the discrete Gaussian over $\lat \setminus \{ \vec0\}$ might be strictly harder than centered DGS. In particular, in Section~\ref{sec:SVPtoDGS}, we show a family of lattices for which $D_{\lat, s}$ almost never returns a $o(\sqrt{n/\log n})$-approximate shortest vector. However, it is easy to see that the discrete Gaussian over the \emph{primitive} lattice vectors or even just over the lattice without $\vec0$ will output the shortest vector with overwhelming probability if the parameter $s$ is sufficiently small. Therefore, both of these sampling problems are actually polynomial-time equivalent to SVP, while we have some evidence to suggest that sampling from $D_{\lat, s}$ is not. Indeed, we know of no application of centered DGS in which non-primitive vectors are actually desirable.}}{}

\subsection{Related work}
\label{sec:related}

\paragraph{DGS algorithms. } There are now many very different algorithms for sampling from the discrete Gaussian. (See Table~\ref{tab:DGS}.) The procedure of \cite{GPV08} (which was originally introduced by Klein in a different context~\cite{Klein00} and was later improved by Brakerski et al.~\cite{BLPRS13}) is a randomized variant of Babai's celebrated nearest plane algorithm \cite{Bab86}. It chooses the coordinates of a lattice vector in a given basis one-by-one by sampling from appropriate shifts of the $n$ one-dimensional Gaussians generated by the Gram-Schmidt orthogonalization of the basis vectors. Peikert showed a similar algorithm that uses the one-dimensional Gaussians generated by the basis vectors themselves instead of their Gram-Schmidt orthogonalizations~\cite{Peikert09}. This yields an elliptical discrete Gaussian, and Peikert convolves this with an elliptical continuous Gaussian in a clever way to obtain a spherical discrete Gaussian. Both of these algorithms are useful for building trapdoor primitives because they can sample from lower parameters if the input basis is shorter.

From our perspective, the algorithms of \cite{Klein00, GPV08,BLPRS13} and \cite{Pei10} can be viewed as reductions from DGS with high parameters $s$ to approximate SVP, where a better approximation factor allows us to sample with a lower parameter $s$ by finding a better basis. And, Regev~\cite{Reg09} explicitly showed a \emph{quantum} reduction from DGS with large $s$ to a different lattice problem. Indeed, many reductions between lattice problems start by sampling vectors from $D_{\lat, s}$ for some large $s$ using one of these algorithms and then using an oracle for some lattice problem to find small combinations of the samples whose average lies in the lattice (e.g.,~\cite{Reg09, MicciancioP13}). One can show that the distribution of the resulting average will be close to $D_{\lat, s'}$ for some $s' < s$ (as long as certain conditions are met).

However, all of the above-mentioned algorithms only work above the smoothing parameter of the lattice because they incur error that depends on \scarequotes{how smooth} the distribution is. Recently, \cite{ADRS15} showed that the averages of pairs of vectors sampled from the centered discrete Gaussian will be distributed \emph{exactly} as discrete Gaussians with a lower parameter, as long as we condition on the averages lying in the lattice. They then showed how to choose such pairs efficiently and proved that this is sufficient to sample from any centered discrete Gaussian in $2^{n+o(n)}$ time---even for parameters $s$ below smoothing. \cite{ADS15} then extended this idea to arbitrary Gaussians (as opposed to just centered Gaussians) with very low parameters $s \gtrsim \dist(\vec{t}, \lat)/2^{n}$. In both cases, the sampler actually outputs exponentially many vectors from the desired distribution.

\paragraph{Sparsification. } The samplers in this work approach discrete Gaussian sampling in a completely different way. (Indeed, the author repeatedly tried and failed to modify the above techniques to work in our context.) Instead, as we described above, we use a new method of sampling based on lattice sparsification. This tool was originally introduced by Khot for the purposes of proving the hardness of approximating SVP~\cite{Khot05svp}. Khot analyzed the behavior of sparsification only on the specific lattices that arose in his reduction, which were cleverly designed to \scarequotes{behave nicely} when sparsified. Later, Dadush and Kun analyzed the behavior of sparsification over general lattices~\cite{DK13} and introduced the idea of adding a random shift to the target in order to obtain deterministic approximation algorithms for CVP in any norm. Dadush, Regev, and Stephens-Davidowitz used a similar algorithm to obtain a reduction from approximate CVP to the same problem with an upper bound on the distance to the lattice (and a slightly smaller approximation factor)~\cite{cvpp}. Our sparsification analysis in the CVP case is most similar to that of~\cite{cvpp}, though our reduction requires tighter analysis. 

However, in the SVP case our analysis is quite different from that of prior work. In particular, we deal explicitly with primitive lattice vectors, which allows us to tightly analyze the behavior of sparsification without a random shift. This seems necessary for studying the distribution of the shortest vector of an arbitrary sparsified lattice, but prior work managed to avoid this by either working with a specific type of lattice or adding a random shift. 

Our use case for sparsification is also novel. In all prior work, sparsification was used to \scarequotes{filter out annoying short vectors, leaving only desirable vectors behind.}  We instead use it specifically to sample from the resulting distribution of the shortest or closest vector in the sparsified lattice. We suspect that this technique will have additional applications.

\paragraph{Dimension-preserving reductions. } More generally, this paper can be considered as part of a long line of work that studies the relationships between various lattice problems under dimension-preserving reductions. Notable examples include~\cite{GMSS99}, which showed that SVP reduces to CVP;~\cite{Micciancio08}, which gave a reduction from SIVP to CVP; and~\cite{LM09}, which showed the equivalence of uSVP, GapSVP, and BDD up to polynomial approximation factors. In particular, this work together with~\cite{Micciancio08} shows that exact SIVP, exact CVP, and DGS are all equivalent under dimension-preserving reductions. (See~\cite{latticereductions} for a summary of such reductions.)

\subsection{Directions for future work}

\paragraph{Centered DGS. } In this work, we completely characterize the complexity of arbitrary discrete Gaussian sampling by showing that it is equivalent to CVP under dimension-preserving reductions. But, the complexity of centered DGS is still unknown. This is therefore the most natural direction for future work. In particular, we show that centered DGS is no harder than SVP (and therefore no harder than NP), but our lower bound only shows that it is at least as hard as $\gamma$-approximate SVP for any $\gamma = \Omega(\sqrt{n/\log n})$. The decision version of SVP is not NP-hard for such high approximation factors unless the polynomial hierarchy collapses, so there is a relatively large gap between our lower and upper bounds. Indeed, for $\gamma = \Omega(\sqrt{n / \log n})$, the decision version of $\gamma$-approximate SVP is known to be in co-AM, and even in SZK~\cite{GG98}.\footnote{The search problem could still potentially be NP-hard for such high approximation factors without violating any widely believed complexity-theoretic conjectures. However, this seems unlikely.} We provide some (relatively weak) evidence to suggest that $\gamma = \Omega(\sqrt{n/\log n})$  is the best achievable approximation factor (see Section~\ref{sec:SVPtoDGS}), and we therefore ask whether centered DGS can be reduced to an easier problem---perhaps even the search variant of a problem in $\mathsf{NP} \cap \mathsf{co}\text{-}\mathsf{AM}$.

A related and arguably much more important question is whether there is an algorithm for centered DGS that is faster than the $2^{n+o(n)}$-time algorithm of~\cite{ADRS15}---perhaps a sampler that outputs only one sample, as opposed to exponentially many. Indeed,~\cite{ADRS15} discuss possible ways to improve their techniques to achieve a roughly $2^{n/2 + o(n)}$-time algorithm for centered DGS, and they make some progress towards this goal. It seems that entirely new techniques would be needed to achieve running times below $2^{n/2}$. Any algorithm with a substantially better constant in the exponent would be the asymptotically fastest algorithm to break nearly all lattice-based cryptographic schemes.

\paragraph{Reductions to approximate lattice problems. } We note that the sampling algorithm of~\cite{Klein00, GPV08, BLPRS13} and many of the DGS subroutines used in hardness proofs can be seen as dimension-preserving reductions from DGS with very high parameters to \emph{approximate} lattice problems. If one simply plugs an exact SVP solver into these reductions, they will still only work for very high parameters. (More specifically, these works can be seen as reducing DGS with $s \gtrsim \gamma \sqrt{\log n} \lambda_n(\lat)$ to $\gamma$-approximate SVP or SIVP.) Our reductions, on the other hand, can handle any parameter but only work with exact solvers. 

We therefore ask if there are better reductions from DGS to \emph{approximate} lattice problems with a better lower bound on the parameter $s$ than the one obtained in~\cite{GPV08, BLPRS13}. Ideally, we would like a smooth trade-off between the approximation factor $\gamma$ and the lower bound on the parameter $s$ that matches our result that works for any $s$ in the exact case when $\gamma = 1$. But, any non-trivial improvement over~\cite{GPV08,BLPRS13} would be a major breakthrough. (A dimension-preserving reduction from DGS with parameter $s \gtrsim \sqrt{(\gamma - 1)/n} \cdot \dist(\vec{t},\lat ) $ to $\gamma$-approximate CVP would show that the two problems are equivalent and therefore completely characterize DGS. Furthermore,~\cite{LiuLM06, cvpp} show that it actually suffices to handle cases when either $\dist(\vec{t}, \lat) \gtrsim \sqrt{\log n/n} \cdot \lambda_1(\lat)$ or $s$ is above the smoothing parameter.) 

Indeed, it is still plausible that we could obtain a dimension-preserving reduction from \emph{centered} DGS to $\gamma$-approximate SVP for some $1 < \gamma \lesssim \sqrt{n/\log n}$. A reduction with $\gamma = \Omega(\sqrt{n/\log n})$ would completely characterize the complexity of centered DGS, but it seems far out of reach. However, any non-trivial $\gamma > 1$ would be quite interesting. In fact, DGS is essentially equivalent to centered DGS above the smoothing parameter. (See, e.g.,~\cite[Section 5.4]{ADRS15}.) So, a result for centered DGS might also advance the study of arbitrary DGS above smoothing.

\section{Preliminaries}

For $\vec{x} \in \R^n$, we write $\length{\vec{x}}$ to represent the $\ell_2$ norm of $\vec{x}$. (Except for the last section, this is the only norm that we consider.) We write $r B_2^n$ to represent the (closed) ball of radius $r$ in $\R^n$, $r B_2^n := \{ \vec{x} \in \R^n\ : \ \length{\vec{x}} \leq r\}$. \full{We will make repeated use of the simple fact that $(1+1/\poly(n))^{1/C} = 1+1/\poly(n)$ for any constant $C$.}{Many of the proofs that are not included in this extended abstract can be found in the full version.}

\subsection{Lattices}

\full{A lattice $\lat\subset \R^n$ is the set of all integer linear combinations of linearly independent vectors $\basis = (\vec{b}_1, \ldots, \vec{b}_n) \in \R^n$. $\basis$ is called a basis of the lattice. As the basis is not unique, we often refer to the lattice itself, as opposed to its representation by a basis.
}{}

We write $\lambda_1(\lat)$ for the length of a shortest non-zero vector in the lattice, and $\lambda_2(\lat)$ is the length of a shortest vector in the lattice that is linearly independent from a vector of length $\lambda_1(\lat)$.
For any $\vec{t} \in \R^n$, we define 
$
\dist(\vec{t}, \lat) := \min_{\vec{x} \in \lat} \length{\vec{x} - \vec{t}}$,
and the covering radius is then $\mu(\lat) := \max_{\vec{t}} \dist(\vec{t}, \lat)$.

\full{We will need basic bounds on $\lambda_1(\lat)$ and $\mu(\lat)$ for rational lattices in terms of the bit length of the basis. (Many of our results are restricted to lattices and targets in $\Q^n$ entirely for the sake of bounds like this. We could instead work over the reals, provided that the chosen representation of real numbers leads to similar bounds.)}{}

\begin{lemma}
\label{lem:lambda1bitlength}
For any lattice $\lat \subset \Q^n$ with basis $\basis = (\vec{b}_1, \ldots, \vec{b}_n)$, let $m$ be a bound on the bit length of $\vec{b}_i$ for all $i$ in the natural representation of rational numbers. Then,
\[
2^{-nm} \leq \lambda_1(\lat) \leq 2^{m}
\; ,
\]
and 
\[
2^{-nm - 1} \leq \mu(\lat) \leq n2^{m}
\; .
\]
\end{lemma}
\full{\begin{proof}
The first upper bound is trivial, as $\lambda_1(\lat) \leq \length{\vec{b}_1} \leq 2^{m}$. For the lower bound, let $q_i$ be a the minimal positive integer such that $q_i \vec{b}_i \in \Z^n$. Note that $q_i \leq 2^{m}$. Then, for any vector $\vec{x} \in \lat$, we have $\vec{x} \cdot \prod_i q_i \in \Z^n$. Therefore, $\lambda_1(\lat) \geq \prod_i q_{i}^{-1} \geq 2^{-n m}$.

Similarly, the lower bound on $\mu(\lat)$ is trivial, as $\mu(\lat) \geq \lambda_1(\lat)/2 \geq 2^{-nm - 1}$. For the upper bound, we have $\mu(\lat) \leq \sum \length{\vec{b}_i} \leq n2^{m}$.
\end{proof}}{}

\full{The following Lemma is due to~\cite{BHW93}.

\begin{lemma}
\label{lem:weakKL}
For any lattice $\lat \subset \R^n$ and $r > 0$, 
\[
|\lat \cap r B_2^n| \leq 1+ \Big(\frac{8r}{\lambda_1(\lat)} \Big)^{n}
\; .
\]
\end{lemma}}{}

\begin{corollary}
\label{cor:ballcountingbitlength}
For any lattice $\lat \subset \Q^n$ with basis $(\vec{b}_1, \ldots, \vec{b}_n)$, $\vec{t} \in \Q^n$, and $r > 0$, let $m$ be a bound on the bit length of the $\vec{b}_i$ for all $i$ in the natural representation of rational numbers. Then,
\[
|(\lat - \vec{t}) \cap rB_2^n| \leq 1 + (2+r)^{\poly(n,m)}
\; .
\]
\end{corollary}
\full{\begin{proof}
It suffices to bound $|\lat \cap (r+\mu(\lat)) B_2^n)|$. The result then follows by applying Lemma~\ref{lem:lambda1bitlength} and Lemma~\ref{lem:weakKL}.
\end{proof}}{}

\subsection{The discrete Gaussian distribution}

For $\vec{x} \in \R^n$ and $s > 0$, we write $\rho_s(\vec{x}) := e^{-\pi \length{\vec{x}}^2/s^2}$. For $A \subset \R^n$, a discrete set, we write $\rho_s(A) := \sum_{\vec{x} \in A} \rho_s(\vec{x})$, and we define the discrete Gaussian distribution over $A$ with parameter $s$, $D_{A,s}$, as the distribution that assigns probability $\rho_s(\vec{x})/\rho_s(A)$ to all $\vec{x} \in A$. When $s = 1$, we omit it and simply write $\rho(\vec{x})$, $D_{\lat}$, etc.

\full{Banaszczyk proved the following two useful bounds on the discrete Gaussian over lattices \cite{banaszczyk}. 

\begin{lemma}
\label{lem:rhoLt}
For any lattice $\lat \subset \R^n$, $s > 0$, and $\vec{t} \in \R^n$, 
\[
\rho_s(\lat - \vec{t}) \geq e^{-\pi \dist(\vec{t}, \lat)^2/s^2}\rho_s(\lat)
\; .
\]
\end{lemma}}{}

\begin{lemma}[{\cite[Lemma 2.13]{cvpp}}]
\label{lem:banaszczyk} 
For any lattice $\lat\subset\R^n$, $s > 0$, $\vec{t} \in \R^n$, and $r \geq1/\sqrt{2\pi}$,
\[
\Pr_{\vec{X} \sim D_{\lat - \vec{t}, s}}[\length{\vec{X}} \geq r s\sqrt{n} ] < \frac{\rho_s(\lat)}{\rho_s(\lat - \vec{t})}\big( \sqrt{2 \pi e r^2} \exp(-\pi r^2) \big)^n
\; .
\]

\end{lemma}

\full{With this, we derive a corollary similar to \cite[Corollary 2.7]{ADS15}.

\begin{corollary}
\label{cor:shiftedbanaszczyk}
For any lattice $\lat \subset \R^n$, $s > 0$, $\vec{t} \in \R^n $, and $r \geq 1/\sqrt{2\pi}$,
\[
\Pr_{\vec{X} \sim D_{\lat - \vec{t}}, s}[\length{\vec{X}}^2 \geq \dist(\vec{t}, \lat)^2 + r^2 s^2 n] < \big( \sqrt{2 \pi e r^{\prime 2}} \exp(-\pi r^{2}) \big)^n
\; ,
\]
where $r' := \sqrt{\dist(\vec{t}, \lat)^2/(s^2 n) + r^2}$. 
In particular, if 
\[r \geq 10\sqrt{\log \big(10 + \dist(\vec{t}, \lat)/(s \sqrt{n})\big)}
\; ,
\] 
then
\[
\Pr_{\vec{X} \sim D_{\lat - \vec{t}}, s}[\length{\vec{X}}^2 \geq \dist(\vec{t}, \lat)^2 + r^2 s^2 n] < e^{-r^2n}
\; .
\]
\end{corollary}
\begin{proof}
Combining the above two lemmas, we have
\begin{align*}
\Pr_{\vec{X} \sim D_{\lat - \vec{t}}, s}[\length{\vec{X}}^2 \geq \dist(\vec{t}, \lat)^2 + r^2 s^2 n] &< e^{\pi \length{\vec{t}}^2/s^2}\cdot \big( \sqrt{2 \pi e r^{\prime 2}} \exp(-\pi r^{\prime 2}) \big)^n \\
&=  \big( \sqrt{2 \pi e r^{\prime 2}} \exp(-\pi r^{2}) \big)^n
\; ,
\end{align*}
as needed.

Now, suppose, $r \geq 10\sqrt{\log \big(10 + \dist(\vec{t}, \lat)/(s \sqrt{n})\big)}$. We consider two cases. First, suppose $\frac{\dist(\vec{t}, \lat)}{s \sqrt{n}} < 1$. Then, we have
$
 r^{\prime 2} < 2r^2 < \frac{e^{r^2}}{2\pi e}
$,
and the result follows. Otherwise, we have
\[
r^{\prime 2} = \frac{\dist(\vec{t}, \lat)^2}{s^2 n} \cdot (1+ r^2 s^2 n/\dist(\vec{t}, \lat)^2)< \frac{\dist(\vec{t}, \lat)^2}{s^2 n} \cdot \exp\Big(\frac{r^2 s^2 n}{\dist(\vec{t}, \lat)^2}\Big)
\;.
\]
So,
\begin{align*}
\sqrt{2 \pi e r^{\prime 2}} \exp(-\pi r^{2}) &< \frac{\dist(\vec{t}, \lat)}{s} \cdot \sqrt{2\pi e/n} \cdot \exp \Big( \frac{r^2 s^2 n}{2\dist(\vec{t}, \lat)^2} - \pi r^2\Big)\\
&< \frac{\dist(\vec{t}, \lat)}{s} \cdot \sqrt{2\pi e/n} \cdot e^{(1/2-\pi) r^2}\\
&< e^{- r^2} \; ,
\end{align*}
as needed.
\end{proof}}
{
\begin{corollary}
\label{cor:shiftedbanaszczyk}
For any lattice $\lat \subset \Q^n$, $s > 0$, $\vec{t} \in \Q^n $, and $r\geq 10\sqrt{\log (10 + \dist(\vec{t}, \lat)/(s \sqrt{n}))}$,
\[
\Pr_{\vec{X} \sim D_{\lat - \vec{t}}, s}[\length{\vec{X}}^2 \geq \dist(\vec{t}, \lat)^2 + r^2 s^2 n] < e^{-r^2n}
\; .
\]
\end{corollary}
}

\full{The following lemma is actually true for \scarequotes{almost all lattices,} in a certain precise sense that is outside the scope of this paper. (See, e.g.,~\cite{siegal45}.)

\begin{lemma}
\label{lem:randomlattice}
For any $n \geq 1$,
there is a lattice $\lat \subset \Q^n$ such that for any $s > 0$, $\rho_s(\lat) \geq 1+s^{n}$ and $\lambda_1(\lat) > \sqrt{n}/10$.
\end{lemma}}
{}

\subsection{Lattice problems}

\full{
\begin{definition}
For any parameter $\gamma \geq 1$, $\gamma$-SVP (the Shortest Vector Problem) is the search problem defined as follows: The input is a basis $\basis$ for a lattice $\lat \subset \Q^n$. The goal is to output a lattice vector $\vec{x}$ with $0 < \length{\vec{x}} \leq \gamma \lambda_1(\lat)$.
\end{definition}

\begin{definition}
For any parameter $\gamma \geq 1$, $\gamma$-CVP (the Closest Vector Problem) is the search problem defined as follows: The input is a basis $\basis$ for a lattice $\lat \subset \Q^n$ and a target vector $\vec{t} \in \Q^n$. The goal is to output a lattice vector $\vec{x}$ with $\length{\vec{x} - \vec{t}} \leq \gamma \dist(\vec{t}, \lat)$.
\end{definition}
}
{
\begin{definition}
SVP (the Shortest Vector Problem) is the search problem defined as follows: The input is a basis $\basis$ for a lattice $\lat \subset \Q^n$. The goal is to output a lattice vector $\vec{x}$ with $ \length{\vec{x}} = \lambda_1(\lat)$.
\end{definition}

\begin{definition}
CVP (the Closest Vector Problem) is the search problem defined as follows: The input is a basis $\basis$ for a lattice $\lat \subset \Q^n$ and a target vector $\vec{t} \in \Q^n$. The goal is to output a lattice vector $\vec{x}$ with $\length{\vec{x} - \vec{t}} = \dist(\vec{t}, \lat)$.
\end{definition}
}

\full{We will mostly be interested in the exact case, when $\gamma = 1$, in which case we simply write SVP and CVP respectively. Note that there may be many shortest lattice vectors or closest lattice vectors to $\vec{t}$.}{}

\begin{definition}
For $\gamma \geq 1$ and $\eps \ge 0$, we say that a distribution $X$ is $(\gamma, \eps)$-close to a distribution $Y$ if there is another distribution $X'$ with the same support as $Y$ such that 
\begin{enumerate}
\item the statistical distance between $X$ and $X'$ is at most $\eps$; and
\item for all $x$ in the support of $Y$, 
$
\Pr[Y = x]/\gamma \leq \Pr[X' = x] \leq \gamma \Pr[Y = x]
$
.
\end{enumerate} 
\end{definition}

\begin{definition}
For any parameters $\eps \geq 0$ and $\gamma \geq 1$, $(\gamma, \eps)$-DGS (the Discrete Gaussian Sampling problem) is defined as follows: 
The input is a basis $\basis$ for a lattice $\lat \subset \Q^n$, a shift $\vec{t} \in \Q^n$, and a (rational) parameter $s > 0$. The goal is to output a vector whose distribution is $(\gamma, \eps)$-close to $D_{\lat - \vec{t}, s}$.
\end{definition}

\begin{definition}
For any parameters $\eps \geq 0$ and $\gamma \geq 1$, $(\gamma, \eps)$-cDGS (the centered Discrete Gaussian Sampling problem) is defined as follows: 
The input is a basis $\basis$ for a lattice $\lat \subset \Q^n$ and a (rational) parameter $s > 0$. The goal is to output a vector whose distribution is $(\gamma, \eps)$-close to $D_{\lat, s}$.
\end{definition}

\full{DGS is typically defined with an additional parameter $\sigma \geq 0$, such that the algorithm only needs to output discrete Gaussian samples if $s > \sigma$. Since both of our reductions achieve $\sigma = 0$, we omit this parameter.}{}

\full{\subsection{Algorithms for one-dimensional Gaussians}

Brakerski, Langlois, Peikert, Regev, and Stehl\'e show how to efficiently sample from the one-dimensional discrete Gaussian $D_{\Z + c, s}$ for any $c \in \R$ and $s > 0$~\cite{BLPRS13}. For completeness, we describe a slightly modified version of their algorithm to sample from $D_{\Z \setminus \{ 0 \}, s}$.

\begin{lemma}
\label{lem:sampleZ}
There is an algorithm that samples from $D_{\Z \setminus \{ 0 \}, s}$ for any $s > 0$ in (expected) polynomial time.
\end{lemma}
\begin{proof}
We describe an algorithm that samples from $D_{\Z^+, s}$, which is clearly sufficient.
Let $Z := e^{-\pi/s^2} + \int_1^\infty e^{-\pi x^2/s^2} {\rm d} x$. The algorithm outputs $1$ with probability $e^{-\pi/s^2}/Z$. Otherwise, it samples $x$ from the one-dimensional continuous Gaussian with parameter $s$ restricted to the interval $(1,\infty)$. Let $y := \ceil{x}$. With probability $e^{-\pi(y^2-x^2)/s^2}$, the algorithm outputs $y$. Otherwise, it repeats.

On a single run of the algorithm, for any integer $z \geq 2$, the probability that the algorithm outputs $z$ is
\[
\frac{1}{Z} \cdot \int_{z-1}^z e^{-\pi x^2/s^2}\cdot e^{-\pi(z^2-x^2)/s^2} {\rm d}x = \frac{e^{-\pi z^2/s^2}}{Z}
\;.
\]
And, the probability that the algorithm outputs $1$ is of course $e^{-\pi/s^2}/Z$. So, the algorithm outputs the correct distribution. 

It remains to bound the expected running time. After a single run, the algorithm outputs an integer with probability
\[
\frac{\rho_s(\Z^+)}{Z} = \frac{\rho_s(\Z^+)}{e^{-\pi/s^2} + \int_1^\infty e^{-\pi x^2/s^2} {\rm d}x } \geq \frac{1}{2}
\; .
\]
It follows that it runs in expected polynomial time.
\end{proof}

Furthermore, we will need to efficiently compute $\rho_s(\Z \setminus \{0 \})$ for arbitrary $s$. Brakerski et al. give a simple algorithm for this problem as well. (Here, we ignore the bit-level concerns of what it means to \scarequotes{efficiently compute} a real number, as this will not be an issue for us.)

\begin{claim}
\label{clm:computerhoZ}
There is an efficient algorithm that computes $\rho_s(\Z \setminus \{0 \})$.
\end{claim}}
{}

\full{
\subsection{\texorpdfstring{Lattice vectors mod $p$ and $\Z_p^n$}{Lattice vectors mod p and Zpn}}
\label{sec:sparseprelims}

Our primary technical tool will be lattice sparsification, in which we consider the sublattice 
\[\lat' := \{ \vec{x} \in \lat\ :\ \inner{\vec{z}, \basis^{-1}\vec{x}} \equiv 0 \bmod p \}
\; ,
\] where $p$ is some prime, $\vec{z} \in \Z_p^n$ is uniformly random, and $\basis$  is a basis of the lattice $\lat \subset \Q^n$. As such, we will need some lemmas concerning the behavior of lattice vectors mod $p\lat $. We first simply note that we can compute $\lat'$ efficiently.

\begin{claim}
\label{clm:efficientsparse}
There is a polynomial-time algorithm that takes as input a basis $\basis$ for a lattice $\lat \subset \R^n$, a number $p \in \Z^+$, and a vector $\vec{z} \in \Z_p^n$ and outputs a basis $\basis'$ for 
\[\lat' := \{ \vec{x} \in \lat\ :\ \inner{\vec{z}, \basis^{-1}\vec{x}} \equiv 0 \bmod p \}
\; .
\]
\end{claim}
\begin{proof}
On input $\basis = (\vec{b}_1,\ldots, \vec{b}_n)$, $p \in \Z^+$, and $\vec{z} = (z_1,\ldots, z_n) \in \Z_p^n$, if $\vec{z} = \vec0$, the algorithm simply outputs $\basis$. Otherwise, we assume without loss of generality that $z_n \neq 0$. The algorithm then computes $\basis^{-T} = (\vec{b}_1^*, \ldots, \vec{b}_n^*)$. It sets 
\[
\hat{\basis} := \Big(\vec{b}_1^*, \ldots, \vec{b}_{n-1}^*, \frac{1}{q}\sum z_i \vec{b}_i^* \Big)
\; .
\] 
Finally, it outputs $\basis' := \hat{\basis}^{-T}$.

A quick computation shows that $\hat{\basis}$ has full rank and that $\basis'$ is indeed a basis for $\lat'$.
\end{proof}

Since we will only be concerned with the coordinates of the vectors mod $p$, it will suffice to work over $\Z_p^n$.

\begin{lemma}
\label{lem:almostindependent}
For any prime $p$ and collection of vectors $\vec{x}, \vec{v}_1,\ldots, \vec{v}_N \in \Z_p^n \setminus \{\vec0 \}$ such that $\vec{x}$ is not a scalar multiple of any of the $\vec{v}_i$, we have
\[
\frac{1}{p} - \frac{N}{p^2} \leq \Pr\big[\inner{\vec{z}, \vec{x}} \equiv 0 \bmod p \text{ and } \inner{\vec{z}, \vec{v}_i} \not\equiv 0 \bmod p \ \forall i \big] \leq \frac{1}{p}
\; ,
\]
where $\vec{z} $ is sampled uniformly at random from $\Z_p^n$.
\end{lemma}
\begin{proof}
For the upper bound, it suffices to note that, since $\vec{x}$ is non-zero, $\inner{\vec{z}, \vec{x}}$ is uniformly distributed over $\Z_p$. Therefore, $\Pr[\inner{\vec{z}, \vec{x}} \equiv 0 \bmod p] = 1/p$.
For the lower bound, note that $A := \{ \vec{y} \in \Z_p^n \ : \ \inner{\vec{y}, \vec{x}} \equiv 0 \bmod p \}$ and $B_i := \{ \vec{y} \in \Z_p^n \ : \ \inner{\vec{y}, \vec{v}_i} \equiv 0 \bmod p \}$ are distinct subspaces of dimension $n-1$. Therefore, $A \cap B_i$ is a subspace of dimension $n-2$ with $p^{n-2}$ elements. Let $B := \bigcup B_i$. It follows that
\begin{align*}
\Pr\big[\inner{\vec{z}, \vec{x}} \equiv 0 \bmod p \text{ and } \inner{\vec{z}, \vec{v}_i} \not\equiv 0 \bmod p \big] &= \frac{|A \setminus B|}{|\Z_p^n |}\\
&\geq \frac{|A| - \sum_i |A \cap B_i|}{|\Z_p^n|}\\
&= \frac{p^{n-1} - N p^{n-2}}{p^n}\\
&= \frac{1}{p} -\frac{N}{p^2}
\; .
\end{align*}
\end{proof}

\begin{corollary}
\label{cor:shiftedindependence}
For any prime $p$, collection of vectors $\vec{v}_1,\ldots, \vec{v}_N \in \Z_p^n$, and $\vec{x} \in \Z_p^n$ with $\vec{x} \neq \vec{v}_i$ for any $i$, we have
\[
\frac{1}{p} - \frac{N}{p^2} - \frac{N}{p^{n-1}} \leq \Pr\big[\inner{\vec{z}, \vec{x} + \vec{c}} \equiv 0 \bmod p \text{ and } \inner{\vec{z}, \vec{v}_i + \vec{c}} \not\equiv 0 \bmod p \ \forall i \big] \leq \frac{1}{p} + \frac{1}{p^{n}}
\; ,
\]
where $\vec{z}$ and $\vec{c} $ are sampled uniformly and independently at random from $\Z_p^n$.
\end{corollary}
\begin{proof}
For the upper bound, it suffices to note that $\Pr[\inner{\vec{z}, \vec{x} + \vec{c}} \equiv 0 \bmod p] \leq \frac{1}{p} + \frac{1}{p^n}$. 

Turning to the lower bound, note that for any $i$, we have $\Pr[\vec{v}_i + \vec{c} = \vec0] = 1/p^n$. By union bound, the probability that $\vec{v}_i + \vec{c} = \vec0$ for any $i$ is at most $N/p^n$. Now, fix $i$, and note that if there exists some $\alpha \in \Z_p \setminus \{1 \}$ such that $\alpha(\vec{v}_i + \vec{c}) = \vec{x} + \vec{c}$, then we must have
\[
\vec{c} = \frac{\alpha \vec{v}_i - \vec{x}}{1-\alpha}
\; .
\]
There are therefore at most $p-1$ values for $\vec{c}$ that satisfy the above---one for each value of $\alpha$. So, the probability that $\vec{c}$ will satisfy the above equation for any $\alpha$ is at most $(p-1)/p^n$. Taking a union bound over all $i$, we see that the probability that $\vec{x} + \vec{c}$ is a multiple of any of the $\vec{v}_i + \vec{c}$ is at most $N(p-1)/p^n$. The result then follows from Lemma~\ref{lem:almostindependent} and union bound.
\end{proof}}
{}

\subsection{Primitive lattice vectors}

\full{For a lattice $\lat \subset \R^n$, we say that $\vec{x} \in \lat$ is non-primitive in $\lat$ if $\vec{x} = k \vec{y}$ for some $\vec{y} \in \lat$ and $k \geq 2$. Otherwise, $\vec{x}$ is primitive in $\lat$. }{}Let $\lat^{\mathrm{prim}}$ be the set of primitive vectors in $\lat$. For a radius $r > 0$, let $\xi(\lat, r) := |\lat^{\mathrm{prim}} \cap r B_2^n|/2 $ be the number of primitive lattice vectors in a (closed) ball of radius $r$ around the origin (counting $\vec{x}$ and $-\vec{x}$ as a single vector).

\full{We will need the following technical lemma, which shows that relatively short primitive vectors cannot be scalar multiples of each other mod $p$.

\begin{lemma}
\label{lem:nogoodnameforthislemma}
For any lattice $\lat \subset \R^n$ with basis $\basis$, suppose $\vec{x}_1, \vec{x}_2 \in \lat$ are primitive with $\vec{x}_1 \neq \pm \vec{x}_2$ and $\length{\vec{x}_1} \geq \length{\vec{x}_2}$ such that 
\[
\basis^{-1}\vec{x}_1 \equiv \alpha \basis^{-1}\vec{x}_2 \bmod p
\; 
\]
for any number $p \geq 100$ and $\alpha \in \Z_p$. Then, $\xi(\lat, \length{\vec{x}_1}) > p/(20\log p)$.
\end{lemma}
\begin{proof}
We assume $\alpha \neq 0$, since otherwise $\vec{x}_1$ is not even primitive.
So, we have that $\vec{x}_1 - q \vec{x}_2 \in p \lat \setminus \{ \vec0\}$ for some integer $q \equiv \alpha \bmod p $ with $0 < |q| \leq p/2$. Let $\vec{y} := (\vec{x}_1 - q \vec{x}_2)/p \in \lat$ and note that $\vec{y}$ is not a multiple of $\vec{x}_2$. It suffices to find at least $\ceil{p/(20 \log p)}$ primitive vectors in the lattice spanned by $\vec{y}$ and $\vec{x}_2$ that are at least as short as $\vec{x}_1$.

We consider two cases. If $q = \pm 1$, then for $i = 0, \ldots, p-1$, the vectors $i\vec{y} + q \vec{x}_2$ are clearly primitive in the lattice spanned by $\vec{y}$ and $\vec{x}_2$, and we have
\[
\length{i\vec{y} + q \vec{x}_2} = \length{i\vec{x}_1 + q(p-i)\vec{x}_2}/p \leq \length{\vec{x}_1}
\; ,
\]
as needed.

Now, suppose $|q| > 1$. Then, for $i = \ceil{p/4},\ldots, \floor{p/2}$, let $k_i$ be an integer such that $|k_i - iq/p| \leq 1/2$ and $0 < |k_i| < i$. (Note that such an integer exists, since $1/2 \leq |iq/p| \leq i/2$). Then,
\begin{align*}
\length{i\vec{y} + k_i\vec{x}_2} &= \length{i \vec{x}_1/p + (k_i -iq/p) \vec{x}_2} \leq \length{\vec{x}_1} 
\;.
\end{align*}
When $i$ is prime, then since $0 < |k_i| < i$, we must have $\gcd(i, k_i) = 1$. Therefore, the vector $i\vec{y} + k_i\vec{x}_2$ must be primitive in the lattice spanned by $\vec{y}$ and $\vec{x}_2$ when $i$ is prime. It follows from a suitable effective version of the Prime Number Theorem that there are at least $\ceil{p/(20 \log p)}$ primes between $\ceil{p/4}$ and $\floor{p/2}$ (see, e.g., \cite{rosser41}), as needed.

\end{proof}

We next show that we can find many primitive lattice vectors in a suitably large ball around $\vec0$.}
{}

\begin{lemma}
\label{lem:notdegenerate}
For any lattice $\lat \subset \R^n$ and radius $r \geq \lambda_2(\lat)$,
\[
\xi(\lat, r) > \frac{\sqrt{r^2-\lambda_2(\lat)^2}}{\lambda_1(\lat)} + \Big\lfloor \frac{r-\lambda_2(\lat)}{\lambda_1(\lat)} \Big\rfloor
\; .
\]
\end{lemma}
\full{\begin{proof}
Let $\vec{v}_1 , \vec{v}_2 \in \lat$ with $\length{\vec{v}_i} = \lambda_i(\lat)$ and $\inner{\vec{v}_1, \vec{v}_2} \geq 0$. Then, for $k = 0, \ldots, \floor{\sqrt{r^2 - \lambda_2(\lat)^2}/\lambda_1(\lat)}$,
\[
\length{\vec{v}_2 - k\vec{v}_1}^2 = \lambda_2(\lat)^2 + k^2 \lambda_1(\lat)^2 - 2k \inner{\vec{v}_1, \vec{v}_2} \leq r^2
\; .
\]
Similarly, for $k = 1, \ldots, \floor{(r-\lambda_2(\lat))/\lambda_1(\lat)}$,
\[
\length{\vec{v}_2 + k \vec{v}_1} \leq \lambda_2(\lat) + k \lambda_1(\lat) \leq r
\]
The result follows by noting that all of these vectors are distinct and primitive in the lattice generated by $\vec{v}_1, \vec{v}_2$ (as is $\vec{v}_1$).
\end{proof}}{}

\full{\subsection{Probability}

We will also need the Chernoff-Hoeffding bound~\cite{hoeffding}.

\begin{lemma}[Chernoff-Hoeffding bound]
\label{lem:chernoff}
Let $X_1, \ldots, X_N $ be independent and identically distributed random variables with $0 \leq X_i \leq 1$ and $\overline{X} := \expect[X_i]$. Then, for $s > 0$
\[
\Pr\Big[N\overline{X} - \sum X_i \geq s \Big] \leq e^{-s^2/N}
\; ,
\]
and
\[
\Pr\Big[\sum X_i - N\overline{X} \geq s \Big] \leq e^{-s^2/N}
\; .
\]
\end{lemma}}
{}

\section{DGS to CVP reduction}
\label{sec:DGStoCVP}

\subsection{Sparsify and shift}

We now present the main sparsification result that we require. In particular Theorem~\ref{thm:shiftedsparsification}\full{ (which is immediate from the work done in Section~\ref{sec:sparseprelims}, and is presented in this form here for the reader's convenience)}{} shows the generic behavior of the sparsification procedure. Proposition~\ref{prop:shiftedsparsifier} then applies the theorem to show how sparsification interacts with a CVP oracle.

\begin{theorem}
\label{thm:shiftedsparsification}
For any lattice $\lat \subset \R^n$ with basis $\basis$, prime $p$, and lattice vectors $\vec{x}, \vec{y}_1,\ldots, \vec{y}_N \in \lat$ such that $\basis^{-1}\vec{x} \not\equiv \basis^{-1}\vec{y}_i \bmod p$ for all $i$, we have
\[
\frac{1}{p} - \frac{N}{p^2} - \frac{N}{p^{n-1}}  \leq \Pr[\inner{\vec{z}, \basis^{-1}\vec{x} + \vec{c}} \equiv \vec0  \text{ and } \inner{\vec{z}, \basis^{-1}\vec{y}_i + \vec{c}} \not\equiv 0 \bmod p \ \forall i] \leq 
 \frac{1}{p} + \frac{1}{p^n} \; ,
\]
where $\vec{z}, \vec{c} \in \Z_p^n$ are chosen uniformly and independently at random.
\end{theorem}
\begin{proof}
Simply apply Corollary~\ref{cor:shiftedindependence} to $\basis^{-1}\vec{x}$ and $\basis^{-1}\vec{y}_i$.
\end{proof}

\begin{proposition}
\label{prop:shiftedsparsifier}
There is a polynomial-time algorithm that takes as input a basis $\basis$ for a lattice $\lat \subset \R^n$  and a prime $p$
and outputs a full-rank sublattice $\lat' \subseteq \lat$ and shift $\vec{w} \in \lat$ such that, for any $\vec{t} \in \R^n$, $\vec{x} \in \lat$ with $N:= |(\lat - \vec{t}) \cap \length{\vec{x} - \vec{t}} \cdot B_2^n| -1  < p$, and any $\problem{CVP}$ oracle,
\[\frac{1}{p} - \frac{N}{p^2} - \frac{N}{p^{n-1}}  \leq \Pr[\problem{CVP}(\vec{t} + \vec{w}, \lat') = \vec{x} + \vec{w}] \leq 
 \frac{1}{p} + \frac{1}{p^n}
\; .
\]
\full{In particular, 
\[
\frac{N}{p} - \frac{N^2}{p^2} - \frac{N^2}{p^{n-1}} \leq \Pr[\dist(\vec{t} + \vec{w}, \lat') \leq \length{\vec{x} - \vec{t}}]  \leq \frac{N}{p} + \frac{N}{p^n}
\; .
\]}{}
\end{proposition}
\full{\begin{proof}
On input $\lat \subset \R^n$ with basis $\basis$ and $p$, the algorithm samples $\vec{z}, \vec{c} \in \Z_p^n$ uniformly and independently at random. It then returns the sublattice
\[
\lat' := \{ \vec{x} \in \lat\ :\ \inner{\vec{z}, \basis^{-1}\vec{x}} \equiv 0 \bmod p\}
\; ,
\]
and the shift $\vec{w} := \basis\vec{c}$.

By Claim~\ref{clm:efficientsparse}, the algorithm can be run in polynomial time.
Let $\vec{y}_1,\ldots, \vec{y}_{N} \in \lat$ be the unique vectors such that $\length{\vec{y}_i - \vec{t}} \leq \length{\vec{x} - \vec{t}}$ with $\vec{y}_i \neq \vec{x}$. 
Note that $\problem{CVP}(\lat', \vec{t} + \vec{w})$ must be $\vec{x} + \vec{w}$ if $\inner{\vec{z}, \basis^{-1} \vec{y}_i + \vec{c}} \not\equiv 0 \bmod p$ for all $i $ \emph{and} $\inner{\vec{z}, \basis^{-1} \vec{x} + \vec{c}} \equiv 0 \bmod p$. We therefore wish to apply Theorem~\ref{thm:shiftedsparsification}, which requires showing that $\basis^{-1}\vec{y}_i \not\equiv \basis^{-1} \vec{x} \bmod p$ for all $i$.

Suppose on the contrary that $\basis^{-1}\vec{y}_i \equiv \basis^{-1} \vec{x} \bmod p$ for some $i$. Then, $\vec{y} := \vec{y}_i - \vec{x} \in p\lat \setminus \{ \vec0\}$, and there are therefore $p+1$ lattice vectors on the line segment between $\vec{y}_i$ and $\vec{x}$ (including the two endpoints). Note that all of these vectors are at least as close to $\vec{t}$ as $\vec{x}$. But, there can be at most $N+1 < p+1$ such vectors, a contradiction. Therefore, we can apply Theorem~\ref{thm:shiftedsparsification}, yielding the result.
\end{proof}}
{}

As a consequence of Proposition~\ref{prop:shiftedsparsifier}, we show that we can use a CVP oracle to sample nearly uniformly from the lattice points in a ball around $\vec{t}$. This relatively straightforward algorithm is the core idea behind our reduction. For simplicity, we provide the algorithm with an estimate of the number of points inside the ball as input. (In the next section, we show how to obtain this estimate using roughly the same techniques.) \full{}{The proof is in the full version.}

\begin{lemma}
\label{lem:shifteduniformsampler}
For any efficiently computable $f(n)$ with $2\leq f(n) \leq \poly(n)$, there is an algorithm with access to a CVP oracle that takes as input a lattice $\lat \subset \Q^n$, shift $\vec{t} \in \Q^n$, radius $r > 0$, and integer $N \geq 1$ and outputs a vector $\vec{y}$ such that, if 
\[
N \leq |\lat \cap (rB_2^n+ \vec{t})| \leq f(n) N
\; ,
\]
then the algorithm runs in expected polynomial time, and for any $\vec{x} \in \lat \cap (rB_2^n+ \vec{t})$, 
\[
\frac{\gamma^{-1}}{ |\lat \cap (rB_2^n+ \vec{t})|} \leq \Pr[\vec{y} = \vec{x}] \leq \frac{\gamma}{|\lat  \cap (rB_2^n+ \vec{t})|}
\; ,
\]
where $\gamma := 1+1/f(n)$.
Furthermore, all of the algorithm's oracle calls are on full-rank sublattices of the input lattice.
\end{lemma}
\full{\begin{proof}
We assume without loss of generality that $n \geq 2$. On input $\lat \subset \Q^n$, $\vec{t} \in \Q^n$, $r >0$, and $N \geq 1$, the algorithm chooses a prime $p$ with $10 f(n)N \leq p \leq 20 f(n) N$ and calls the procedure from Proposition~\ref{prop:shiftedsparsifier} on input $\lat$ and $p$, receiving as output a sublattice $\lat' \subseteq \lat$ and a shift $\vec{w} \in \lat$. It then calls its CVP oracle on input $\lat'$ and $\vec{t} + \vec{w}$, receiving as output $\vec{y}'$. If $\length{\vec{y}' - \vec{w} - \vec{t}} \leq r$, it outputs $\vec{y} := \vec{y}' - \vec{w}$. Otherwise, it repeats.

From Proposition~\ref{prop:shiftedsparsifier}, we have that, after a single run of the algorithm,
\[
\frac{1}{\sqrt{\gamma} \cdot p}  \leq \frac{1}{p} - \frac{N}{p^2} - \frac{N}{p^{n-1}} \leq \Pr[\vec{y}' = \vec{x} + \vec{w}] \leq \frac{1}{p} + \frac{1}{p^n} \leq \frac{\sqrt{\gamma}}{p}
\; .
\]
Correctness follows immediately. Furthermore, note that the reduction outputs something on each run with probability at least 
$
\frac{N}{\sqrt{\gamma} f(n) p} \geq \frac{1}{100 f(n)^2}$. So, in particular, the expected number of runs is polynomial in $n$. It is clear that a single run takes polynomial time, and the result follows.
\end{proof}}{}

\subsection{Counting the lattice vectors in a ball}

We now show how to use the sparsification algorithm to approximate the number of lattice points in a ball, given access to a CVP oracle. \full{We will use this both to instantiate the procedure from Lemma~\ref{lem:shifteduniformsampler} and directly in our DGS sampling procedure.}{The proof is in the full version.}

\begin{definition}
For any parameter $\gamma \geq 1$, $\gamma$-GapVCP (the Vector Counting Problem) is the promise problem defined as follows: the input is a lattice $\lat \subset \Q^n$ (represented by a basis), shift $\vec{t} \in \Q^n$, radius $r > 0$, and an integer $N \geq 1$. It is a NO instance if 
$|(\lat - \vec{t}) \cap r B_2^n| \leq N$ and a YES instance if $|(\lat - \vec{t}) \cap r B_2^n| > \gamma N$.
\end{definition}

\begin{theorem}
\label{thm:counter}
For any efficiently computable function $f(n)$ with $1\leq f(n) \leq \poly(n)$, there is a polynomial-time reduction from $\gamma$-GapVCP to CVP where $\gamma := 1+1/f(n)$. The reduction \full{preserves dimension and }{}only calls the CVP oracle on sublattices of the input lattice.
\end{theorem}
\full{\begin{proof}
We assume without loss of generality that $n \geq 20$ and $f(n) \geq 20$. On input a lattice $\lat \subset \Q^n$ with basis $\basis$, target $\vec{t} \in \Q^n$, $r > 0$, and $N \geq 1$, the reduction behaves as follows.  First, it finds a prime $p$ with $200f(n)N \leq p \leq 400f(n) N$. Then, for $i = 1, \ldots, \ell := \ceil{100f(n)^2p^2 /N^2}$, the reduction calls the procedure from Proposition~\ref{prop:shiftedsparsifier} on $\lat$, $\vec{t}$, and $p$. It receives as output $\lat_i$ and $\vec{w}_i$. It then calls the CVP oracle on $\lat_i$ and $\vec{t} + \vec{w}_i$, receiving as output a vector whose distance from $\vec{t} + \vec{w}_i$ is $r_i$. Finally, it returns yes if $r \leq r_i$ for all but at most $\ell N/p + 2\sqrt{\ell}$ values of $r_i$ and no otherwise.

It is clear that the reduction runs in polynomial time. Now, suppose $|\lat \cap (rB_2^n+ \vec{t})| \leq N$. By Proposition~\ref{prop:shiftedsparsifier}, we have that for each $i$,
\[
\Pr[r_i \leq r] \leq \frac{N}{p} + \frac{N}{p^n} < \frac{N}{p} + \frac{1}{2\sqrt{\ell}}
\; .
\] 
Then, applying the Chernoff-Hoeffding bound (Lemma~\ref{lem:chernoff}), we have
\[
\Pr[|\{ i\ :\ r_i \leq r\}| > \ell N/p + 2\sqrt{\ell}] < 1/e
\; .
\]
So, the reduction returns the correct answer in this case with probability at least $1-1/e$.

On the other hand, suppose that $|\lat \cap (rB_2^n+ \vec{t})| > \gamma N$. Using the lower bound in Proposition~\ref{prop:shiftedsparsifier},
\[
\Pr[r_i \leq r] \geq \frac{\gamma N}{p} - \frac{\gamma^2 N^2}{p^2} - \frac{\gamma^2 N^2}{p^{n-1}} > \frac{N}{p} + \frac{5}{\sqrt{\ell}}
\; .
\]
Applying the Chernoff-Hoeffding bound again, we have
\[
\Pr[|\{ i\ :\ r_i \leq r\}| \leq \ell N/p + 2\sqrt{\ell}] < 1/e
\;,
\]
as needed.
\end{proof}}
{}

\subsection{The DGS algorithm}

\begin{theorem}
\label{thm:DGStoCVP}
For any efficiently computable function $f(n)$ with $1 \leq f(n) \leq \poly(n)$, there exists an (expected) polynomial-time reduction from $(\gamma, \eps)$-DGS to CVP, where $\eps := 2^{-f(n)}$ and $\gamma := 1+1/f(n)$. The reduction preserves dimension and only calls the CVP oracle on full-rank sublattices of the input lattice.
\end{theorem}
\begin{proof}
We assume without loss of generality that $n \geq 5$ and $s = 1$. (If $s \neq 1$, we can simply rescale the lattice.) On input $\lat \subset \Q^n$ and $\vec{t} \in \Q^n$, the reduction behaves as follows. It first calls its CVP oracle to compute $d := \dist(\vec{t}, \lat)$. For $i = 0,\ldots, \ell := \ceil{100n^2 f(n) \log(10+ d)} $, let $r_i := \sqrt{d^2 +i/(10f(n))}$.  
For each $i$, the reduction uses its CVP oracle together with the procedure given in Theorem~\ref{thm:counter} to compute $N_i$ such that $\gamma^{-1/10} \cdot  |(\lat - \vec{t}) \cap r_iB_2^n| \leq N_i \leq  |(\lat - \vec{t}) \cap r_iB_2^n|$.

Let $w_\ell := e^{-\pi r_\ell^2}$, and for $i = 0, \ldots, \ell - 1$, let $w_i := e^{-\pi r_i^2}-e^{-\pi r_{i+1}^2}$. Let $W := \sum_{i=0}^\ell N_i w_i$. The reduction then chooses an index $0 \leq k \leq \ell$, from the distribution that assigns to index $i$ probability $N_i w_i/W$. 
It then runs the procedure from Lemma~\ref{lem:shifteduniformsampler} with input $\lat$, $\vec{t}$, $r_k$, and $N_k$, receiving as output a vector $\vec{y} \in (\lat - \vec{t}) \cap r_k B_2^n$ whose distribution is $(\gamma^{1/10}, 0)$-close to the uniform distribution over $(\lat - \vec{t}) \cap r_k B_2^n$. It then simply returns $\vec{y}$.

\full{To see that the reduction runs in polynomial time, first note that Lemma~\ref{lem:lambda1bitlength} implies that $\ell$ is polynomial in the length of the input. Similarly, Corollary~\ref{cor:ballcountingbitlength} implies that the $N_i$ have bit lengths polynomial in the length of the input. It follows that the reduction runs in expected polynomial time.

We now prove correctness.}{It is clear that the reduction runs in polynomial time.}
Let $A := (\lat - \vec{t}) \cap r_\ell B_2^n$ be the support of $\vec{y}$. By Corollary~\ref{cor:shiftedbanaszczyk}, $D_{A}$ is within statistical distance $\eps$ of $D_{\lat - \vec{t}}$\full{, so it suffices to show that the output of the reduction is $(\gamma,0)$-close to $D_A$. In order to show this, it suffices to show that, for any $\vec{x} \in A$, $\Pr[\vec{y} = \vec{x}]$ is proportional to $\rho(\vec{x})$, up to a factor of $\gamma^{\pm 1/2}$.}{.}
Note that
\begin{equation}
\label{eq:probformula}
\Pr[\vec{y} = \vec{x}] = \frac{1}{W} \sum_{i\ :\ r_i \geq \length{\vec{x} - \vec{t}}} w_iN_i \cdot \Pr[\vec{y} = \vec{x} \ |\ k = i] 
\; .
\end{equation}
For any $i$ such that $\vec{x} \in (\lat - \vec{t}) \cap r_i B_2^n$, by Lemma~\ref{lem:shifteduniformsampler} we have that 
\[
\frac{\gamma^{-1/5}}{ N_i} \leq \frac{\gamma^{-1/10}}{|(\lat - \vec{t}) \cap r_i B_2^n|} \leq \Pr[\vec{y} = \vec{x} \ | \ k = i] \leq \frac{\gamma^{1/10}}{|(\lat - \vec{t}) \cap r_i B_2^n|} \leq \frac{\gamma^{1/10}}{N_i} 
\; .
\]
\full{Let $j$ be minimal such that $\vec{x} \in (\lat - \vec{t}) \cap r_j B_2^n$. Plugging in the upper bound to Eq.~\eqref{eq:probformula}, we have
\[
\Pr[\vec{y} = \vec{x}] \leq \frac{\gamma^{1/10}}{W}\cdot \sum_{i \geq j} w_i = \frac{\gamma^{1/10}}{W} \cdot e^{-\pi r_j^2} \leq \frac{\sqrt{\gamma}}{W} \cdot \rho(\vec{x})
\; .
\]
A nearly identical computation shows that $\Pr[\vec{y} = \vec{x}] \geq \rho(\vec{x})/(\sqrt{\gamma} W)$, as needed.}{Plugging these bounds into Eq.~\eqref{eq:probformula} gives the result.}
\end{proof}
\section{Centered DGS to SVP reduction}
\label{sec:DGStoSVP}
\subsection{Sparsification}

\full{Since we are now interested in the SVP case, we can no longer handle the shifts used in Theorem~\ref{thm:shiftedsparsification} and Proposition~\ref{prop:shiftedsparsifier} (neither the input shift $\vec{t}$ nor the output shifts $\vec{w}$ and $\vec{c}$). As a result, we are forced to consider the effect of sparsification on primitive vectors only, which requires new analysis. }{}Recall that $\xi(\lat, r) := |\lat_{\mathrm{prim}} \cap r B_2^n|/2$ is the number of primitive lattice vectors in a ball of radius $r$ (counting $\pm \vec{x}$ as a single vector).

\begin{theorem}
\label{thm:sparsification}
For any lattice $\lat \subset \R^n$ with basis $\basis$, primitive lattice vectors $\vec{y}_0, \vec{y}_1,\ldots, \vec{y}_N \in \lat_{\mathrm{prim}}$ with $\vec{y}_i \neq \pm \vec{y}_0$ for all $i > 0$, and prime $p \geq 101$, if $\xi(\lat, \length{\vec{y}_i}) \leq p/(20 \log p)$ for all $i$, then
\[
\frac{1}{p} - \frac{N}{p^2} \leq \Pr\big[\inner{\vec{z}, \basis^{-1}\vec{y}_0} \equiv 0 \bmod p \text{ and } \inner{\vec{z}, \basis^{-1}\vec{y}_i} \not\equiv 0 \bmod p \ \forall i > 0\big] \leq \frac{1}{p} \; ,
\]
where $\vec{z} \in \Z_p^n$ is chosen uniformly at random.
\end{theorem}
\begin{proof}
Let $\vec{v}_i := \basis^{-1}\vec{y}_i$. By Lemma~\ref{lem:nogoodnameforthislemma}, we have that $\vec{v}_0 $ is not a scalar multiple of $\vec{v}_i$ mod $p$ for any $i > 0$. The result then follows from Lemma~\ref{lem:almostindependent}.
\end{proof}

\full{}{The proof of the next result is in the full version.}

\begin{proposition}
\label{prop:sparsifier}
There is a polynomial-time algorithm that takes as input a basis $\basis$ for a lattice $\lat \subset \R^n$ and a prime $p \geq 101$ 
and outputs a full-rank sublattice $\lat' \subseteq \lat$ such that for every $\vec{x} \in \lat$ with
$N:= \xi(\lat, \length{\vec{x}}) - 1 \leq p/(20 \log p)$ and $\lambda_1(\lat) > \length{\vec{x}}/p$, we have that for any SVP oracle,
\[
\frac{1}{p} - \frac{N}{p^2} \leq \Pr[\problem{SVP}(\lat') = \pm \vec{x}] \leq \frac{1}{p}
\; .
\]
\full{In particular, 
\[
\frac{N}{p} - \frac{N^2}{p^2} \leq \Pr
\big[\lambda_1(\lat') \leq \length{\vec{x}}\big]  \leq \frac{N}{p}
\; .
\]}{}
\end{proposition}
\full{\begin{proof}
On input $\lat \subset \R^n$ with basis $\basis$ and $p$, the algorithm samples $\vec{z} \in \Z_p^n$ uniformly at random. It then returns the sublattice
\[
\lat' := \{ \vec{x} \in \lat\ :\ \inner{\vec{z}, \basis^{-1}\vec{x}} \equiv 0 \bmod p\}
\; .
\]

It is clear that the algorithm runs in polynomial time. 
Since $\Pr[\vec{x} \in \lat'] = 1/p$, the upper bound on the probability is immediate as well. 

For the lower bound, let $\vec{y}_0,\ldots, \vec{y}_{N} \in \lat_{\mathrm{prim}}$ such that $\length{\vec{y}_i} \leq \length{\vec{x}}$, $\vec{y}_i \neq \pm \vec{y}_j$, and $\vec{y}_0 := \vec{x}$. Let $\vec{v}_i := \basis^{-1}\vec{y}_i$. Note that, if $\vec{v}_0 \in \lat'$ and $\vec{v}_i \notin \lat'$ for $i > 0$, then $\problem{SVP}(\lat') = \pm \vec{x}$. (Here, we have used the fact that $\lambda_1(\lat) > \length{\vec{x}}/p$.)
The result then follows from Theorem~\ref{thm:sparsification}.
\end{proof}}
{}

\begin{lemma}
\label{lem:uniformsampler}
For any efficiently computable $f(n)$ with $2\leq f(n) \leq \poly(n)$, there is an (expected) polynomial-time algorithm with access to a SVP oracle that takes as input a lattice $\lat \subset \Q^n$, radius $r > 0$, and integer $N \geq 1$ and outputs a vector $\vec{y} \in \lat$ such that, if 
$N \leq \xi(\lat, r) \leq f(n) N$ and $\lambda_1(\lat) > r/(f(n)\xi(\lat, r))$,
then for any $\vec{x} \in \lat^{\mathsf{prim}} \cap r B_2^n$, 
\[
\frac{\gamma^{-1}}{\xi(\lat, r)} \leq \Pr[\vec{y} = \pm \vec{x}] \leq \frac{\gamma}{\xi(\lat, r)}
\; ,
\]
where $\gamma := 1+f(n)$.
Furthermore, the algorithm \full{preserves dimension and }{}only calls its oracle on full-rank sublattices of $\lat$.
\end{lemma}
\full{\begin{proof}
We assume without loss of generality that $n \geq 10$. On input $\lat \subset \Q^n$, $r >0$, and $N \geq 1$, the algorithm chooses a prime $p$ with $100 f(n)N \log(10f(n)N) \leq p \leq 200 f(n) N \log(10 f(n)N)$ and calls the algorithm from Proposition~\ref{prop:sparsifier} on input $\lat$ and $p$, receiving as output a sublattice $\lat' \subset \lat$. It then calls its SVP oracle on input $\lat'$, receiving as output $\vec{y}$. If $\length{\vec{y}} \leq r$, it outputs $\vec{y}$. Otherwise, it repeats.

From Proposition~\ref{prop:sparsifier}, we have that, after a single run of the algorithm
\[
\frac{\gamma^{-1/2}}{p}  \leq \frac{1}{p} - \frac{N}{p^2} - \frac{N}{p^{n-1}} \leq \Pr[\vec{y} = \pm \vec{x}] \leq \frac{1}{p}
\; .
\]
Correctness follows immediately. Furthermore, note that the algorithm terminates after a given run with probability at least $\gamma^{-1/2} N/(f(n)p) \geq 1/(1000 f(n)^2 \log(Nf(n)))$. By Corollary~\ref{cor:ballcountingbitlength}, $\log(N)$ is polynomial in the length of the input. So, in particular, the expected number of runs is polynomial in the length of the input. It is clear that a single run takes polynomial time, and the result follows.
\end{proof}}{}

\subsection{Counting the primitive lattice vectors in a ball around the origin}

\begin{definition}
For any parameters $\beta \geq 0$, $\gamma \geq 1$, $(\beta, \gamma)$-GapPVCP (the Primitive Vector Counting Problem) is the promise problem defined as follows: the input is a lattice $\lat \subset \Q^n$ (represented by a basis), radius $r > 0$, and an integer $N \geq 1$. It is a NO instance if 
$\xi(\lat,  r) \leq N$ or if $\lambda_1(\lat) <  \beta r/ N$ and a YES instance if $\xi(\lat,  r) > \gamma N$.
\end{definition}

\full{Intuitively, the condition that $\lambda_1(\lat) <  \beta r/ N$ handles the degenerate case in which there are many non-primitive vectors that may \scarequotes{hide} the primitive vectors in the lattice. It is not clear that this should be treated as a degenerate case in general, but it is clear that our methods fail in this case.
}
{In the full version, we prove the following.}

\begin{theorem}
\label{thm:primcounter}
For any efficiently computable $f(n)$ with $1 \leq f(n) \leq \poly(n)$, there is a polynomial-time reduction from $(\beta, \gamma)$-GapPVCP to SVP where $\beta := 1/f(n)$ and $\gamma := 1+1/f(n)$. The reduction preserves dimension and only calls the SVP oracle on sublattices of the input lattice.
\end{theorem}
\full{\begin{proof}
On input $\lat \subset \Q^n$ with basis $\basis$, $r > 0$, and $N \geq 1$, the reduction behaves as follows. It first calls its SVP oracle on $\lat$ to compute $\lambda_1(\lat)$. If $\lambda_1(\lat) > r$ or $\lambda_1(\lat) < \beta r/N$, it returns no. The reduction then finds a prime $p$ with $200f(n) N\log(10f(n)N) \leq p \leq 400 f(n)N \log(10f(n)N)$, and for $i = 1, \ldots, \ell := \ceil{100f(n)^2p^2/N^2}$, it calls the procedure from Proposition~\ref{prop:sparsifier} on $\lat$ and $p$, receiving as output $\lat_i$. It then calls the SVP oracle on each $\lat_i$, receiving as output a vector of length $r_i$. Finally, it returns yes if $r \leq r_i$ for all but at most $\ell N/p + 2\sqrt{\ell}$ values of $r_i$ and no otherwise.

It is clear that the reduction runs in polynomial time. We assume $\lambda_1(\lat) \geq \beta r/N > r/p$ (since otherwise the reduction clearly outputs the correct answer). 

Suppose $m := \xi(\lat,  r) \leq N$. By Proposition~\ref{prop:sparsifier}, we have
$
\Pr[r_i \leq r] \leq \frac{m}{p} \leq \frac{N}{p}
$,
for each $i$.
Applying the Chernoff-Hoeffding bound (Lemma~\ref{lem:chernoff}), we have
\[
\Pr\Big[|\{i\ :\  r_i \leq r\}| >  \frac{N\ell }{p} + 2\sqrt{\ell} \Big] < 1/e
 \; .
\] 
So, the reduction returns the correct answer in this case with probability at least $1-1/e$.

Now, suppose $\xi(\lat, r) > \gamma N$. We again apply Proposition~\ref{prop:sparsifier} to obtain 
\[
\Pr[r_i \leq r] \geq \frac{\gamma N}{p} - \frac{\gamma^2 N^2}{p^2} > \frac{N}{p} + \frac{5}{\sqrt{\ell}}
\]
Applying the Chernoff-Hoeffding bound again, we have
\[
\Pr\Big[|\{i\ :\  r_i \leq r\}| \leq \frac{N\ell}{p} + 2\sqrt{\ell}  \Big] < 1/e
 \; .
\] 
The result follows.
\end{proof}}{}

\subsection{The centered DGS algorithm}

\begin{theorem}
\label{thm:DGStoSVP}
For any efficiently computable function $f(n)$ with $1 \leq f(n) \leq \poly(n)$, there is an (expected) polynomial-time reduction from $(\gamma, \eps)$-cDGS to SVP, where $\eps := 2^{-f(n)}$ and $\gamma := 1+1/f(n)$. The reduction preserves dimension and only calls the SVP oracle on sublattices of the input lattice.
\end{theorem}
\begin{proof}
We assume without loss of generality that $s = 1$. (If $s \neq 1$, we can simply scale the lattice.) On input $\lat \subset \Q^n$, the reduction behaves as follows. First, it computes $\lambda_1(\lat)$ using its SVP oracle. For $i = 0,\ldots, \ell := \ceil{200n^2 f(n)^2} $, let $r_i := \sqrt{\lambda_1(\lat)^2 +i/(100n f(n))}$. For each $i$, the reduction uses its SVP oracle together with the procedure given in Theorem~\ref{thm:primcounter} to compute $N_i$ such that 
\begin{equation}
\label{eq:NiapproxXi}
\gamma^{-1/10}\cdot  \xi(\lat, r_i) \leq N_i \leq  \xi(\lat, r_i)
\; ,
\end{equation} 
or $N_i := 1$ if $\lambda_1(\lat) < r_i/(100 n^2 f(n) \xi(\lat, r_i))$. Let $w_{\ell} := \rho_{1/r_\ell}(\Z \setminus \{ \vec0 \})$, and for $i = 0, \ldots, \ell - 1$, let $w_i := \rho_{1/r_i}(\Z\setminus \{ 0 \}) - \rho_{1/r_{i+1}}(\Z\setminus \{ 0 \})$. (\full{Claim~\ref{clm:computerhoZ}}{\cite{BLPRS13}} shows one way to compute $w_i$ efficiently.) 

Let $W := \sum_{i=0}^\ell N_i w_i$. Then, the reduction outputs $\vec0$ with probability $1/(1+W)$. Otherwise, it chooses an index $0 \leq k \leq \ell$, assigning to each index $i$ probability $N_i w_i/W$. If $N_k > 1$, the reduction then calls the procedure from Lemma~\ref{lem:uniformsampler} on input $\lat$, $r_k$, and $N_k$, receiving as output a vector $\vec{x} \in \lat^{\mathsf{prim}}$ that is distributed uniformly over $\lat^{\mathsf{prim}} \cap r_k B_2^n$, up to a factor of $\gamma^{\pm 1/10}$. If $N_k = 1$, the reduction simply sets $\vec{x} = \problem{SVP}(\lat)$. Finally, it \full{uses the procedure from Lemma~\ref{lem:sampleZ} to sample}{samples} an integer $z$ from $D_{\Z \setminus \{ 0 \}, 1/\length{\vec{x}}}$ and returns $\bar{\vec{x}} := z \cdot  \vec{x}$. \full{}{(\cite{BLPRS13} shows how to sample such an integer efficiently.)}

\full{First, we note that the reduction runs in expected polynomial time. In particular, the $N_i$ have polynomial bit length by Corollary~\ref{cor:ballcountingbitlength}, and the various subprocedures have expected running times that are polynomial in the length of their input.

We now prove correctness.}{It is clear that the reduction runs in polynomial time.} Let $\lat^\dagger$  be the set of all points that are integer multiples of a lattice vector whose length is at most $r_\ell > \sqrt{n f(n)} $. By Lemma~\ref{lem:banaszczyk}, it suffices to consider the distribution $D_{\lat^\dagger}$, as this is within statistical distance $\eps$ of $D_{\lat}$. Then,
\begin{align*}
\rho(\lat^\dagger \setminus \{\vec0 \}) = \sum_{\vec{y} \in \lat^\dagger \setminus \{ \vec0 \}} \rho(\vec{y})
= \sum_{\vec{y} \in \lat^{\mathsf{prim}} \cap \sqrt{n} B_2^n} \rho_{1/\length{\vec{y}}}(\Z \setminus \{ 0\})
\; .
\end{align*}
A quick computation shows that for any $\vec{y}$ with $r_{i-1} \leq \length{\vec{y}} \leq r_i$, we have
\[
\rho_{1/r_i}(\Z \setminus \{ \vec0\}) \leq \rho_{1/\length{\vec{y}}}(\Z \setminus \{ 0\}) \leq \gamma^{1/10} \cdot \rho_{1/r_i}(\Z \setminus \{ 0\})
\; .
\]
Recalling the definition of the $w_i$, it follows that
\[
\sum_{i=0}^\ell  \xi(\lat, r_i) w_i \leq \rho(\lat^\dagger \setminus \{\vec0 \}) \leq \gamma^{1/10} \cdot \sum_{i=0}^\ell \xi(\lat, r_i) w_i
\; .
\]

Now, we would like to say that $N_i \approx \xi(\lat, r_i)$, as in Eq.~\eqref{eq:NiapproxXi}. This is of course true by definition \emph{except} when $N_i = 1$ and $\xi(\lat, r_i) > 1$, i.e., when $\lambda_1(\lat) < r_i/(100 n^2 f(n) \xi(\lat, r_i))$ and $\lambda_2(\lat) \leq r_i$. But, in this case, a quick computation together with Lemma~\ref{lem:notdegenerate} shows that $\xi(\lat, r_{i+1}) > 1/(100 n f(n)\lambda_1(\lat))$, and therefore $N_{j}$ satisfies Eq.~\eqref{eq:NiapproxXi} for all $j > i$. (In other words, the $N_i$ can only be \scarequotes{wrong} for at most one value of $i$.) It follows that, for any $i < \ell$, we have
\[
\gamma^{-1/5}\cdot \sum_{j \geq i} \xi(r_j, \lat_j)w_j \leq   \sum_{j \geq i} N_j w_j \leq   \sum_{j \geq i} \xi(r_j, \lat_j)w_j
\;. 
\] 
(The case $N_\ell = 1$ can be handled separately. Correctness in this case follows essentially immediately from Lemma~\ref{lem:banaszczyk}.)
Putting everything together, we have that
\[
\gamma^{-1/5} \cdot \rho(\lat^\dagger \setminus \{\vec0 \}) \leq W \leq \gamma^{1/5} \cdot \rho(\lat^\dagger \setminus \{\vec0 \})
\; .
\]
So, in particular, the probability that the reduction outputs $\vec0$ is \full{$1/(1+W)$, which is a good approximation to the correct probability of $1/\rho(\lat^\dagger)$}{$1/(1+W) \approx 1/\rho(\lat^\dagger)$, as needed}.

Now, for any $\vec{y} \in \lat^{\mathsf{prim}}$, it follows from Lemma~\ref{lem:uniformsampler} and the argument above that 
\begin{equation}
\label{eq:prbound}
\gamma^{-1/2} \cdot \frac{\rho_{1/\length{\vec{y}}}(\Z \setminus \{0\})}{\rho(\lat^\dagger)}\leq \Pr[\vec{x} = \pm \vec{y}] \leq \gamma^{1/2} \cdot \frac{\rho_{1/\length{\vec{y}}}(\Z \setminus \{0\})}{\rho(\lat^\dagger)}
\; .
\end{equation}
Finally, for any $\vec{w} \in \lat^\dagger \setminus \{0 \}$, let $\vec{y}$ be one of the two primitive lattice vectors that are scalar multiples of $\vec{w}$, and let $\bar{z}$ such that $\vec{w} = \bar{z} \vec{y}$. Then,
\full{\begin{align*}
\Pr[\bar{\vec{x}} = \vec{w}] &= \Pr[\vec{x} = \pm \vec{y}] \cdot \Pr[z = \bar{z}\ |\ \vec{x} = \pm \vec{y}]\\
&= \Pr[\vec{x} = \pm \vec{y}] \cdot \frac{\rho(\vec{w})}{\rho_{1/\length{\vec{y}}}(\Z \setminus \{0\})}
\end{align*}}
{\[
\Pr[\bar{\vec{x}} = \vec{w}] = \Pr[\vec{x} = \pm \vec{y}] \cdot \Pr[z = \bar{z}\ |\ \vec{x} = \pm \vec{y}]= \Pr[\vec{x} = \pm \vec{y}] \cdot \frac{\rho(\vec{w})}{\rho_{1/\length{\vec{y}}}(\Z \setminus \{0\})}
\]}
The result follows from plugging the above equation into Eq.~\eqref{eq:prbound}.
\end{proof}
\section{Sampling from other distributions}
\label{sec:other}
We note that our reductions from Sections~\ref{sec:DGStoCVP} and~\ref{sec:DGStoSVP} do not use any unique properties of the discrete Gaussian distribution or of the $\ell_2$ norm. Above, we focused on this particular case because it has so many applications, while other distributions on lattices seem to be of much less interest. In this section, we show that a much larger class of sampling problems can be reduced to CVP in various different norms. 

First, we show that the sparsification result in Proposition~\ref{prop:shiftedsparsifier} naturally extends to arbitrary norms $K$. In particular, for any norm $K$, we can use a CVP oracle in norm $K$ to sample (nearly) uniformly from the lattice points in a $K$-ball. (See below for the definitions.) We can naturally extend this to any distribution that can be efficiently written as the weighted average of uniform distributions over the lattice points in $K$-balls. For example, this will be enough to show how to use a CVP oracle in the $\ell_q$ norm to sample from the natural $\ell_q$ generalization of the discrete Gaussian, which assigns to $\vec{x} \in \lat - \vec{t}$ probability proportional $e^{- \|\vec{x} \|_q^q}$, where $ \|\vec{x} \|_q := (\sum |x_i|^q)^{1/q}$ for $1 \leq q < \infty$ is the $\ell_q$ norm. 

Below, we make this precise. For simplicity, we will not worry about the more difficult analogous problem of reducing sampling from centered distributions to SVP.

\subsection{Arbitrary distributions and norms}

Recall that any norm $\|\cdot \|_K$ over $\R^n$ is uniquely represented by a compact symmetric convex body with non-empty interior $K \subset \R^n$, its unit ball. The norm itself is then simply
\[
\| \vec{x} \|_{K} := \min \{ r \ : \ \vec{x} \in r K \}
\; .
\]
(Since we are interested in asymptotics, we formally identify $K := (K_1,K_2,\ldots )$ with a sequence of such bodies with $K_n \subset \R^n$, but we will ignore such details.) A $K$-ball with center $\vec{c}$ and radius $r$ is $rK + \vec{c}$, the set of all points within distance $r$ of $\vec{c}$ in the norm $\|\cdot \|_K$.

We define the general problem that interests us below, together with the natural generalization of CVP to arbitrary norms.

\begin{definition}
For any $\gamma \geq 1$, $\eps > 0$, and function $\chi$ mapping a shifted lattice $\lat - \vec{t}$ to a distribution over $\lat - \vec{t}$, the sampling problem $(\gamma, \eps)\text{-}\problem{LSP}_{\chi}$ (the Lattice Sampling Problem) is defined as follows: The input is (a basis of) a lattice $\lat \subset \Q^n$ and a shift $\vec{t} \in \Q^n$. The goal is to output a vector whose distribution is $(\gamma, \eps)$-close to $\chi(\lat - \vec{t})$.
\end{definition}

\begin{definition}
For any norm $\|\cdot \|_K$, the search problem $\problem{CVP}_K$ (the Closest Vector Problem in norm $K$) is defined as follows: The input is (a basis of) a lattice $\lat \subset \Q^n$ and a target vector $\vec{t} \in \Q^n$. The goal is to output a lattice vector $\vec{x}$ such that $\length{\vec{x} - \vec{t}}_K$ is minimal.
\end{definition}

\subsection{Sparsify, shift, count, and sample}

We now observe that Proposition~\ref{prop:shiftedsparsifier} generalizes to arbitrary norms. (One can simply check that the proof of Proposition~\ref{prop:shiftedsparsifier} does not use any special properties of the $\ell_2$ norm.)

\begin{proposition}
\label{prop:shiftedsparsifier-ell_q}
There is a polynomial-time algorithm that takes as input a basis $\basis$ for a lattice $\lat \subset \R^n$  and a prime $p$
and outputs a sublattice $\lat' \subseteq \lat$ and shift $\vec{w} \in \lat$ such that, for any norm $\| \cdot \|_K$, $\vec{t} \in \R^n$, $\vec{x} \in \lat$ with $N:= |(\lat - \vec{t}) \cap \length{\vec{x} - \vec{t}} \cdot K| < p$, and any $\problem{CVP}_K$ oracle,
\[\frac{1}{p} - \frac{N}{p^2} - \frac{N}{p^{n-1}}  \leq \Pr[\problem{CVP}_K(\vec{t} + \vec{w}, \lat') = \vec{x} + \vec{w}] \leq 
 \frac{1}{p} + \frac{1}{p^n}
\; .
\]
\end{proposition}

And, from this, we obtain a generalization of Lemma~\ref{lem:shifteduniformsampler} and Theorem~\ref{thm:counter}.

\begin{definition}
For any parameter $\gamma \geq 1$ and norm $\|\cdot\|_K$, $\gamma\text{-GapVCP}_K$ (the Vector Counting Problem in norm $K$) is the promise problem defined as follows: the input is (a basis of) a lattice $\lat \subset \Q^n$, shift $\vec{t} \in \Q^n$, radius $r > 0$, and an integer $N \geq 1$. It is a NO instance if 
$|(\lat - \vec{t}) \cap r K| \leq N$ and a YES instance if $|(\lat - \vec{t}) \cap rK| > \gamma N$.
\end{definition}

\begin{theorem}
\label{thm:counter_q}
For any efficiently computable norm $\| \cdot \|_K$ and efficiently computable function $f(n)$ with $2\leq f(n) \leq \poly(n)$, there is a polynomial-time reduction from $\gamma\text{-GapVCP}_K$ to $\problem{CVP}_K$, where $\gamma := 1+1/f(n)$. Furthermore, there is an (expected) polynomial-time reduction from $(\gamma, 0)\text{-}\problem{LSP}_{\chi}$ to $\problem{CVP}_K$, where $\chi(\lat - \vec{t})$ is the uniform distribution on $(\lat- \vec{t}) \cap K$ (or $\chi$ is constant on $-\vec{t}$ if $(\lat - \vec{t}) \cap K$ is empty). Both reductions preserve dimension and only make calls to the $\problem{CVP}_K$ oracle on sublattices of the input lattice.
\end{theorem}

\subsection{Sufficiently \scarequotes{nice} distributions and the sampler}

Recall that the sampling algorithm from Theorem~\ref{thm:DGStoCVP} works by computing a finite sequence of balls $B_0,\ldots, B_\ell$ such that the discrete Gaussian distribution is $(\gamma, \eps)$-close to a weighted average of the uniform distributions over these balls. This motivates the following definition and theorem. 

\begin{definition}
\label{def:ball-decomposable}
For a norm $K$, $\gamma = \gamma(n) \geq 1$, and $\eps = \eps(n) > 0$, we say that a function $\chi$ that maps a shifted lattice $\lat - \vec{t}$ to a distribution over $\lat - \vec{t}$ is $(\gamma, \eps, K)$-ball decomposable if it is $(\gamma, \eps)$-close to a weighted average of uniform distributions over the lattice points inside $K$-balls, and these balls and weightings can be computed efficiently with access to a $\problem{CVP}_K$ oracle.
\end{definition}

\begin{theorem}
For any efficiently computable norm $K$, $\gamma = \gamma(n) \geq 1$, and $\eps = \eps(n) > 0$, if $\chi$ is $(\gamma, \eps, K)$-ball decomposable, then for any efficiently computable function $2 \leq f(n) \leq \poly(n)$, there is a polynomial-time reduction from $(\gamma', \eps)\text{-}\problem{LSP}_{\chi}$ to $\problem{CVP}_K$, where $\gamma' := (1+1/f(n)) \gamma $. The reduction preserves dimension and only calls its oracle on sublattices of the input lattice.
\end{theorem}
\begin{proof}
On input $\lat \subset \Q^n$ and $\vec{t} \in \Q^n$, the reduction first calls the procedure guaranteed by Definition~\ref{def:ball-decomposable} to obtain a sequence of $K$-balls $B_0,\ldots, B_\ell$ and weights $w_0,\ldots, w_\ell$. It then selects an index $i$ with probability $w_i$. Finally, it uses the sampling procedure from Theorem~\ref{thm:counter_q} to sample a vector that is $(\gamma^{1/10}, 0)$-close to uniform over $|(\lat - \vec{t} ) \cap B_i|$ and outputs the result.

It is clear that the reduction runs in polynomial time. Correctness follows from the correctness of the various subprocedures and some simple calculations.
\end{proof}

\begin{corollary}
For any efficiently computable function $2 \leq f(n) \leq \poly(n)$ and constant $1 \leq q < \infty$, there is an efficient reduction from $(\gamma, \eps)\text{-}\problem{LSP}_{\chi_q}$ to $\problem{CVP}_{\ell_q}$, where $\gamma := 1+1/f(n)$, $\eps := e^{-f(n)}$, and $\chi_q(\lat - \vec{t})$ is the distribution that assigns to each $\vec{x} \in \lat - \vec{t}$ probability proportional to $e^{-\|\vec{x}\|_q^q}$.
\end{corollary}
\begin{proof}
It suffices to show that $\chi_q$ is $(\sqrt{\gamma}, \eps, \ell_q)$-ball decomposable, i.e., that there is an efficient algorithm with access to a $\problem{CVP}_q$ oracle that outputs balls and weights as in Definition~\ref{def:ball-decomposable}. The algorithm first computes $d := \min_{\vec{y} \in \lat} \| \vec{y} - \vec{t} \|_q$ using its $\problem{CVP}_q$ oracle. For $i = 0, \ldots, \ell := 100 n^q f(n)^{q+1}$, let $r_i := (d^q + i/(10f(n)))^{1/q}$, $\vec{c}_i := \vec0$, and $B_i := r_i K + \vec{c}_i$. Let $\hat{w}_\ell := e^{-r_\ell^q}$, and for $0 \leq i < \ell$, let $\hat{w}_i := e^{-r_i^q} - e^{-r_{i+1}^q}$. The algorithm then uses the counting procedure from Theorem~\ref{thm:counter_q} to approximate $|(\lat - \vec{t}) \cap B_i| = |(\lat - \vec{t} - \vec{c}) \cap r_iK|$ up to a factor of $\gamma^{1/10}$, receiving as output $N_i$. Finally, let $w_i := N_i \hat{w}_i$. The algorithm then simply outputs the $B_i$ and $w_i$.

A simple calculation shows that this is a valid $(\sqrt{\gamma}, \eps, \ell_q)$-ball decomposition of $\chi_q$.
\end{proof}

\full{\section{\texorpdfstring{$\sqrt{n/\log n}$}{sqrt(n/log n)}-SVP to centered DGS reduction and a lower bound}
\label{sec:SVPtoDGS}

It is an immediate consequence of Lemma~\ref{lem:banaszczyk} that $O(\sqrt{n})$-SVP reduces to DGS. In fact, we can do a bit better.\footnote{Interestingly, \cite{ADRS15} achieves nearly identical parameters in a different context with a very different algorithm. They work over the dual and only solve the decisional variant of $\gamma$-SVP. Though they are interested in exponential-time algorithms, it is easy to see that their approach yields a polynomial-time reduction from (the decisional variant of) $\gamma$-SVP to DGS for any $\gamma = \Omega(\sqrt{n/\log n})$. See~\cite[Theorem 6.5]{ADRS15}. Their reduction only requires samples above the smoothing parameter, which is in some sense the reason that they only solve the decisional variant of SVP.}

\begin{proposition}
For any efficiently computable function $10 \leq f(n) \leq \poly(n)$, there is a polynomial-time reduction from $\gamma$-SVP to $(f, \eps)$-DGS, where $\gamma := 10\sqrt{ \frac{ n}{\log f(n)}}$, and $\eps := 1/f(n)$. The reduction only calls the oracle on the input lattice.
\end{proposition}
\begin{proof}
We assume without loss of generality that $n$ is large enough so that $f(n) < 2^{n-1}$. On input $\lat \subset \Q^n$, the reduction behaves as follows. Let $d_{\min}, d_{\max} > 0$ such that $d_{\min} \leq \lambda_1(\lat) \leq d_{\max}$ such that the bit lengths of $d_{\min}$ and $d_{\max}$ are polynomially bounded. (E.g., we can take $d_{\min}$ and $d_{\max}$ to be the values guaranteed by Lemma~\ref{lem:lambda1bitlength}.) For $i = 0, \ldots, 100n^2 \ceil{\log (d_{\max}/d_{\min})}$, let 
\[
s_i := (1+1/n^2)^i \cdot \frac{d_{\min}}{\sqrt{ \log f(n)}}
\; .
\] 
The reduction calls the DGS oracle on input $\lat$ and $s_i$ for each $i$, $\ceil{100nf(n)^{2}}$ times. It then returns the shortest resulting non-zero vector.

It is clear that the reduction runs in polynomial time. Let $i$ such that $ s_{i-1} \leq 
10\sqrt{\frac{1}{\log f(n)}} \cdot 
\lambda_1(\lat) < s_i$. Note that 
\[
\Pr_{\vec{X} \sim D_{\lat, s_i}}[\vec{X} = \vec0] < \frac{1}{1+4/f(n)} < 1-2/f(n)
\; .
\]
By Lemma~\ref{lem:banaszczyk}, 
\[
\Pr_{\vec{X} \sim D_{\lat, s_i}}\big[\length{\vec{X}} > \gamma \cdot \lambda_1(\lat)\big] \leq  \Pr_{\vec{X} \sim D_{\lat, s_i}}[\length{\vec{X}} > s_i\sqrt{n}] < 2^{-n}
\; .
\]
Therefore, if the samples were truly from $D_{\lat, s_i}$, each would be a valid approximation with probability at least $2/f(n)-2^{-n}$. It follows that each sample from the DGS oracle is a valid approximation with probability at least $1/f(n)^2-2^{-n}/f(n) > 1/(2f(n)^2)$, and the result follows.
\end{proof}

We now show a lower bound on the length of non-zero discrete Gaussian vectors. In particular, for any approximation factor $\gamma = o(\sqrt{n/\log n})$, we show a lattice (technically, a family of lattices indexed by the dimension $n$) such that the probability that $D_{\lat, s}$ yields a $\gamma$-approximate shortest vector is negligible for any $s$. This shows that any efficient reduction from $\gamma$-SVP to DGS with $\gamma = o(\sqrt{n/\log n})$ must output a vector not returned by the DGS oracle and/or make DGS calls on a lattice other than the input lattice.

\begin{theorem}
\label{thm:SVPtoDGS}
For any sufficiently large $n$ and $2 < t < \sqrt{n}/10$, there exists a lattice $\lat \subset \Q^n$ with $\lambda_1(\lat) = t$ such that for any $s > 0$, 
\[
\Pr_{\vec{X} \sim D_{\lat, s}}[0 < \length{\vec{X}} \leq \sqrt{n}/10] < e^{-t^2}
\; .
\]
In particular, for any $t = \omega(\sqrt{\log n})$, $D_{\lat, s}$ will yield a $\sqrt{n}/(10t)$-approximate shortest vector with at most negligible probability.
\end{theorem}
\begin{proof}
Fix $n$. 
Let $\lat' \subset \Q^{n-1}$ be an $(n-1)$-dimensional lattice with $\rho_s(\lat') \geq 1+s^{n-1}$ and $\lambda_1(\lat) > \sqrt{n-1}/10$, as promised by Lemma~\ref{lem:randomlattice}. Then, let $\lat := \lat' \oplus t\Z$ be the lattice obtained by \scarequotes{appending} a vector of length $t$ to $\lat'$. Note that the only vectors of length at most $\sqrt{n-1}/10$ in $\lat$ are those that are multiples of the \scarequotes{appended} vector. So, 
\[
\Pr_{\vec{X} \sim D_{\lat, s}}[0 < \length{\vec{X}} \leq \sqrt{n-1}/10] \leq \frac{\rho_s(t\Z \setminus \{ \vec0\})}{\rho_s(\lat')} \leq \frac{\rho_{s/t}(\Z \setminus \{ \vec0\})}{1+s^{n-1}}
\; .
\]
Now, if $s \leq t$, then the numerator is less than $e^{-t^2}$. If $s > t$, then we have
\[
\frac{\rho_{s/t}(\Z \setminus \{ \vec0\})}{1+s^{n-1}} < \frac{s}{1+s^{n-1}} < \frac{1}{s^{n/2}} < \frac{1}{t^{n/2}} < e^{-t^2}
\; ,
\]
where we have used the fact that $\rho_{s'}(\Z \setminus \{0\}) < s'$, and the fact that $2 < t < \sqrt{n}/10$.
\end{proof}
\section{CVP to DGS reduction}
\label{sec:CVPtoDGS}
For completeness, we give a simple reduction from CVP to DGS. It suffices to find a parameter $s$ that is small enough so that the weight of a closest vector to the target is much larger than the weight of all non-closest vectors. The only slightly non-trivial observation necessary is that we can take $s$ large enough that it still has polynomial bit length.

\begin{proposition}
\label{prop:CVPtoDGS}
For any efficiently computable function $2 \leq f(n) \leq \poly(n)$, there is a polynomial-time reduction from CVP to $(f, \eps)$-DGS where $\eps := 1-\frac{1}{f(n)}$. The reduction succeeds with probability at least $1/(2f(n)^2)$ and only makes one oracle call on $\lat - \vec{t}$ where $\lat$ is the input lattice and $\vec{t}$ is the input target.
\end{proposition}
\begin{proof}
On input $\lat \subset \Q^n$ and $\vec{t} \in \Q^n$, the reduction behaves as follows. Let $q \geq 1$ with polynomial bit length such that $\lat \subseteq \Z^n/q$ and $\vec{t} \in \Z^n/q$. Let $d$ be the upper bound on $\mu(\lat)$ guaranteed by Lemma~\ref{lem:lambda1bitlength} (which in particular has polynomial bit length), and let $s := (100f(n)n q  \log (10+d))^{-1}$. The reduction simply samples $\vec{y}$ from $D_{\lat - \vec{t}, s}$ and returns $\vec{y} + \vec{t} \in \lat$.

It is clear that the reduction runs in polynomial time. Note that for any point $\vec{x} \in \lat$ that is not a closest point to $\vec{t}$, we must have $\length{\vec{x} - \vec{t}}^2 \geq \dist(\vec{t}, \lat)^2 + 1/q^2$.  By Corollary~\ref{cor:shiftedbanaszczyk}, we have
\[
\Pr_{\vec{X} \sim D_{\lat - \vec{t}, s}}[\length{\vec{X}}^2 \geq \dist(\vec{t}, \lat)^2 + 1/q^2] < e^{-1/(ns^2q^2)} < e^{-2f(n)n}
\; .
\]
Therefore, any distribution within statistical distance $\eps$ of $D_{\lat - \vec{t}, s} + \vec{t}$ must output a closest point with probability at least $1/f(n) - e^{-2f(n)n} > 1/(2f(n))$. It follows that the oracle outputs a closest point with probability at least $1/(2f(n)^2)$, as needed.
\end{proof}

\begin{corollary}
CVP is equivalent to DGS under polynomial-time, dimension-preserving reductions.
\end{corollary}
\begin{proof}
Combine Theorem~\ref{thm:DGStoCVP} with Proposition~\ref{prop:CVPtoDGS}.
\end{proof}

}{}

\full{\section*{Acknowledgments}

I would like to thank Divesh Aggarwal, Daniel Dadush, and Oded Regev for many enlightening discussions and for their helpful comments on early drafts of this work; Daniele Micciancio for finding a bug in an earlier version of Proposition~\ref{prop:CVPtoDGS}; and the SODA reviewers for their very helpful and thorough reviews.}{}

\full{\bibliographystyle{alpha}}{\bibliographystyle{abbrv}}
\newcommand{\etalchar}[1]{$^{#1}$}

\end{document}